\documentclass[1p,final]{elsarticle}
\usepackage{amsfonts,natbib,pslatex}
\usepackage{latexsym,amsmath,amsthm,amssymb,amsxtra,mathrsfs,times,CJK}
\usepackage{color}
\usepackage{amsfonts,color}
\usepackage{amssymb,amsmath,latexsym}
\usepackage{amssymb}
\usepackage{makecell}
\usepackage{amsfonts}
\usepackage{indentfirst}
\usepackage{amsmath}
\usepackage{amsthm}
\usepackage{amssymb}
\usepackage{color}
\usepackage{graphicx}
\usepackage[para]{threeparttable}
\usepackage{booktabs}
\usepackage{caption}
\usepackage{makecell}
\usepackage{diagbox}
\usepackage{multirow}
\usepackage{epstopdf}
\usepackage{subfigure}
\usepackage{booktabs}
\usepackage{hyperref}

\usepackage[para]{threeparttable}

\newcommand{\ord}{{\rm ord}}

\newcommand{\lcm}{{\rm lcm}}

\newtheorem{theorem}{Theorem}[section]
\newtheorem{lemma}[theorem]{Lemma}
\newtheorem{corollary}[theorem]{Corollary}

\newtheorem{example}[theorem]{Example}
\newtheorem{Proposition}[theorem]{Proposition}
\newtheorem{remark}[theorem]{Remark}

\begin{document}

	\begin{frontmatter}
		
		
		
		\title{MDS and AMDS symbol-pair codes are constructed from  repeated-root codes \tnoteref{fn1}}
		

		\author[SWJTU]{Xiuxin Tang}
		\ead{XiuxinTang@my.swjtu.edu.cn}
		\author[SWJTU]{Rong Luo\corref{cor1}}
		\ead{luorong@swjtu.edu.cn}
		%
		%
		\cortext[cor1]{Corresponding author}
		\address[SWJTU]{School of Mathematics, Southwest Jiaotong University, Chengdu, 610031, China}
		\begin{abstract}
		Symbol-pair codes introduced by Cassuto and Blaum in 2010 are designed to protect against the pair errors in symbol-pair read channels. One of the central themes in symbol-error correction is the construction of maximal distance separable (MDS)  symbol-pair codes that possess the largest possible pair-error correcting performance. In this paper, we construct more general generator polynomials for two classes of MDS symbol-pair codes with code length $lp$. Based on repeated-root cyclic codes, we derive all MDS symbol-pair codes of length  $3p$, when the degree of the generator polynomials is no more than 10. We also give two new classes of (almost maximal distance separable) AMDS symbol-pair codes with the length $lp$ or $4p$ by virtue of repeated-root cyclic codes. For length $3p$, we derive all AMDS symbol-pair codes, when the degree of the generator polynomials is less than 10.
		 The main results are obtained by determining the solutions of certain equations over finite fields.
		\end{abstract}
		
		
	\end{frontmatter}

	\section{Introduction}
	
With the development of high-density data storage technologies, a new coding framework called symbol-pair codes was proposed by Cassuto and Blaum (2010) to combat symbol-pair-errors over symbol-pair read channels in \cite{Cassuto-2010-SIT}. The seminal works \cite{Cassuto-2010-SIT}--\cite{Cassuto-2011-SIT} have established relationships between the minimum Hamming distance of an error-correcting code and the minimum symbol-pair distance, have found methods for code constructions and decoding, and have obtained lower and upper bounds on code size. If a code $\mathcal{C}$ over ${\Bbb F}_p^{{n}}$ of length $n$ contains $M$ elements and has the minimum symbol-pair distance ${d_p}$, then $\mathcal{C}$ is referred as an ${\left( {n,M,{d_p}} \right)_p}$ symbol-pair code. Finding symbol-pair codes with high symbol-pair error correcting performance has become a great challenge in the theory. Therefore, we wish to obtain symbol-pair codes that possess minimum symbol-pair distance as large as possible.

In 2012, Chee et al. \cite{Chee-2012-SIT} established the Singleton-type bound on symbol-pair codes. Similar to the classical codes, the symbol-pair codes meeting the Singleton-type bound are called MDS symbol-pair codes. Due to the optimality, MDS symbol-pair codes are the most useful and interesting symbol-pair codes. Many researchers attempted to obtain MDS symbol-pair codes by different mathematical tools.  Constructing MDS symbol-pair codes is thus of significance in theory and practice. However, not much work has been done on determining the exact values of symbol-pair distances of constacyclic codes as it is a very complicated and difficult task in general. 

In 2013, Chee et al. \cite{Chee-2013-SIT} obtained some MDS symbol-pair codes from classical MDS codes.
Between 2015 and 2018, researchers \cite{Kai-2015-TIT}\cite{Li-2017-DCC}\cite{Kai-2018-NCT} constructed some MDS symbol-pair codes with minmum symbol-pair distance 5 and 6 from constacyclic codes.
In 2017, Chen et al.\cite{Chen-2017-CIT} proposed to construct MDS symbol-pair codes by repeated-cyclic codes and obtained some length $3p$ MDS symbol-pair codes with symbol-pair distance 6 to 8 and MDS ${\left( {lp,5} \right)_p}$  symbol-pair codes. In the next few years, MDS symbol-pair codes with some new parameters were found  by using repeated-root codes over ${\Bbb F}_p$. 
In 2018,  Kai et al.\cite{Kai-2018-NCT} proved the existence of an MDS ${(4p,7)} _p $  sympol-pair code. In the same year, Dinh et al.\cite{Dinh-2018-TIT} presented all MDS symbol-pair codes of prime power lengths by repeated-root constacyclic codes.
In 2019, Zhao \cite{Zhao-2019-Doc} gave an MDS ${(lp,6)} _p $  sympol-pair code.
In 2020, Dinh et al. \cite{Dinh-2020-TIT} constructed a families of MDS sympol-pair codes with length $2p$, where $g_{i,j} =  {{{(x - 1)}^i}{{(x + 1 )}^j}} $ is expressed as the generator polynomials of these codes and ${\left| {i - j} \right| \leqslant 2}$.
In 2021, Ma et al. \cite{Ma-2022-DCC} obtained MDS $ (3p,10)_{p} $  and MDS $ (3p,12)_{p} $ sympol-pair codes.

Inspired by these works, in order to obtain longer and more flexible symbol-pair codeword length, as well as a larger minimum symbol-pair distance, this paper proves that there are more general generator polynomials for MDS $ {(lp,6)}_p $ and MDS $ {(lp,5)}_p $ symbol-pair codes by using repeated-cyclic codes. For length $n=3p$, this paper gives all MDS symbol-pair codes from repeated-root cyclic codes $\mathcal{C}_{r_{1}r_{2}r_{3}}$, when the degree of the generator polynomial $g_{r_{1}r_{2}r_{3}}(x)$ is no more than 10, i.e $\deg( g(x))\le 10$. Furthermore, the parameters of AMDS $ {(4p,8)}_p $ and AMDS $ {(lp,7)}_p $  symbol-pair codes are obtained by using repeated-cyclic codes. For length $3p$, this paper also obtains all AMDS symbol-pair codes from repeated-root cyclic codes $\mathcal{C}_{r_{1}r_{2}r_{3}}$, when $\deg( g(x))\le 10$.

The remainder of this paper is organized as follows. In Section 2, we introduce some basic notations and results on symbol-pair codes. In Section 3, on the basis of repeated-root cyclic codes, we derive some new classes of MDS symbol-pair codes and new classes of AMDS symbol-pair codes. In section 4, we conclude the paper.

\section{Preliminaries}

In this section, we introduce some notations and auxiliary tools on symbol-pair codes, which will be used to prove our main results in the sequel.
Let $${\bf{\textit{\textbf{x}}}} = \left( {{x_0},{x_1}, \cdots ,{x_{n - 1}}} \right)$$ be a vector in ${\Bbb F}_p^{n}$. Then the symbol-pair read vector of $\textit{\textbf{x}}$ is
$$\pi \left( {\bf{\textit{\textbf{x}}}} \right) = \left[ {\left( {{x_0},{x_1}} \right),\left( {{x_1},{x_2}} \right), \cdots ,\left( {{x_{n - 1}},{x_0}} \right)} \right].$$
Similar to the Hamming weight ${\omega _H}\left( {\bf{\textit{\textbf{x}}}} \right)$ and Hamming distance $ {D_H}\left( \bf{\textit{\textbf{x}}},\bf{\textit{\textbf{y}}} \right) $. The symbol-pair weight ${\omega _p}\left( {\bf{\textit{\textbf{x}}}} \right)$ of the symbol-pair vector ${\bf{\textit{\textbf{x}}}}$ is defined as
$${\omega _p}\left( {\bf{\textbf{\textit{x}}}} \right) = \left| {\left\{ {i\left| {\left( {{x_i},{x_{i + 1}}} \right) \ne \left( {0,0} \right)} \right.} \right\}} \right|.$$
The symbol-pair distance $ {D_p}\left( \bf{\textit{\textbf{x}}},\bf{\textit{\textbf{y}}} \right) $ between any two vectors $\bf{\textit{\textbf{x}}},\bf{\textit{\textbf{y}}}$ is
$${D_p}\left( \bf{\textit{\textbf{x}}},\bf{\textit{\textbf{y}}} \right) = \left| {\left\{ {\left. i \right|\left( {{x_i},{x_{i + 1}}} \right) \ne \left( {{y_i},{y_{i + 1}}} \right)} \right\}} \right|.$$
The minimum symbol-pair distance of a code $\mathcal{C}$ is
$${d_p} = \min \left\{ { {{D_p}\left( \bf{\textit{\textbf{x}}},\bf{\textit{\textbf{y}}} \right)\left|\; \bf{\textit{\textbf{x}}},\bf{\textit{\textbf{y}}} \in \mathcal{C} \right.} } \right\},$$
and we denote $ (n,k,d_{p})_{p} $ a symbol-pair code   with length $ n $, dimension $ k $ and minimum symbol-pair distance $ d_{p} $ over  ${\Bbb F}_p$. 

In this paper, we always regard the codeword \textbf{c} in $\mathcal{C}$ as the corresponding polynomial $c(x)$. We donote that ${\Bbb F}_p$ and ${\Bbb F}_q$ are finite fields, where $p$ is an odd prime and $q=p^m$.
The following lemmas will be applied in our later proofs.

In contrast to the classical error-correcting codes, the size of a symbol-pair code satisfies the following Singleton bound. The symbol-pair code achieving the Singleton bound is called a maximum distance separable (MDS) symbol-pair code.

\begin{lemma}\label{lemma 3.4 }	
	
	(\cite{Chee-2012-SIT})If $\mathcal{C}$ is a symbol-pair code with length $n$ and minimum symbol-pair distance ${d_p}$ over ${\Bbb F}_q$, we call an  $ (n,d_p)_{p} $ symbol-pair code of size $ q^{n - {d_p} + 2} $ maximum distance separable (MDS)  and an  $ (n,d_p)_{p} $ symbol-pair code of size $ q^{n - {d_p} + 1} $ almost maximum distance separable (AMDS) for $ q\geqslant 2 $.
	
\end{lemma}

\begin{lemma}\label{lemma 3.2 }
	(\cite{Chee-2013-SIT}) Let $q \geqslant 2$ and $2 \leqslant {d_p} \leqslant n$. If $\mathcal{C}$ is a symbol-pair code with length $n$ and minimum symbol-pair distance ${d_p}$ over ${\Bbb F}_q$, then $\left|\mathcal{C}\right| \leqslant {q^{n - {d_p} + 2}}.$
\end{lemma}

The next lemma is known MDS symbol-pair codes, which will be used by the later proof of Theorme\ref{theorem 3.10}. 
\begin{lemma}\label{lemma 3.5 }
	(\cite{Kai-2018-NCT}) Let $n=4p$ with $p \equiv 3\left({\bmod~4} \right)$. If $\mathcal{C}$ is the cyclic code in ${{\Bbb F}_p}\left[ x \right]/\left\langle {{x^n} - 1} \right\rangle $ generated by $$g\left( x \right) = {\left( {x - 1} \right)^3}\left( {{x^2} + 1} \right),$$ then we have $d_H=4$ and $d_p=7$. $\mathcal{C}$ is an MDS $ (4p,7)_{p} $ symbol-pair code.
\end{lemma}

In some cases, the bound of minmum symbol-pair distance can be improved. 
\begin{lemma}\label{lemma 3.3 }
	(\cite{ Chen-2017-CIT}) Let $\mathcal{C}$ be an $\left[ {n,k,{d_H}} \right]$ constacyclic code over ${{\Bbb F}_q}$ with ${\text{2}} \leqslant {{\text{d}}_{\text{H}}}\leqslant{\text{n}}$. Then we have the following ${d_p}\left( {\mathcal{C}} \right) \geqslant {d_H} + 2$ if and only if $\mathcal{C}$ is not an MDS code.
\end{lemma}

The next lemma is known MDS symbol-pair code, which will be used by the later proof of Proposition \ref{Proposition C_{421}}. 

\begin{lemma}\label{lemma 3.6 }
	(\cite{Chen-2017-CIT}) Let $n=3p$ with $p \equiv 1\left( {\bmod~3} \right)$. If $\mathcal{C}$ is the cyclic code in ${{\Bbb F}_p}\left[ x \right]/\left\langle {{x^n} - 1} \right\rangle $ generated by 
	\[g(x) = {(x - 1)^3}{(x - \omega )^2}(x - {\omega ^2}),\]
	then we have $\mathcal{C}$ is an MDS $ (3p,8)_{p} $ symbol-pair code, where $\omega$ is a primitive $3$-th  root of unity in ${{\Bbb F}_p}$.
	
\end{lemma}

Researchers constructed some MDS symbol-pair codes with minimum symbol-pair distance 6 by repeated-root cyclic codes.  The following three lemmas are known MDS symbol-pair codes with minimum symbol-pair distance 6. These three lemmas will be used as several parts of Theorem \ref{theorem 3.1}.

\begin{lemma} \label{lemma 3.7 }
	
	\cite{Dinh-2018-TIT} $ \mathcal{C} $ is an MDS symbol-pair code, when $\left( {r_1},{r_2},{r_3}\right) =\left(4,0,0 \right)$ is satisfied. 
	$ \mathcal{C} $ is not an MDS symbol-pair code, when we have $\left( {r_1},{r_2},{r_3}\right) =\left(0,4,0 \right) ,\left(0,0,4 \right) $.		
\end{lemma}

\begin{lemma} \label{lemma 3.8 }
	\cite{Zhao-2019-Doc} ${\mathcal{C}} = \left\langle {{{(x - 1)}^3}{{(x - \omega )}}} \right\rangle $ is an MDS ${\left( {lp,{\text{\;}}6} \right)_p}$ symbol-pair code with minmum pair-distance $d_p=6$ over ${{\Bbb F}_p}$, where $\omega$ is a primitive $ l$-th  root of unity in ${{\Bbb F}_p}$.
	
\end{lemma}

\begin{lemma} \label{lemma 3.9 }
	\cite{Dinh-2020-TIT} ${\mathcal{C}} = \left\langle {{{(x - 1)}^i}{{(x + 1 )}^j}} \right\rangle $ is an MDS symbol-pair code with minmum symbol-pair distance $d_p=6$ over ${{\Bbb F}_p}$, where ${\left| {i - j} \right| \leqslant 2}$ and $i,j\le p-1$.		
\end{lemma}

The following lemma gives the method to calculate of minimum Hamming distance about repeated-cyclic codes.
\begin{lemma}\label{lemma 3.1}
	(\cite{Castagnoli-1991-OIT}) Let $\mathcal{C}$ be a repeated-root cyclic code of length $ lp^{e} $ over ${{\Bbb F}_q}$ generated by $ g(x) = \prod {{m_i}^{{e_i}}\left( x \right)}$ , where $ l $ and $ e $ are positive integers with $ \gcd (l, p) = 1 $. Then we have $${d_H}(\mathcal{C}) = \min \left\{ {{P_t} \cdot {d_H}\left( {{{\mathcal{\overline C}}_t}} \right)\left| {0 \le t \le l{p^e}} \right.} \right\},$$
	where $ {P_t} = {\omega _H}\left( {{{\left( {x - 1} \right)}^t}} \right) $ and $ \mathcal{\overline C}_{t} =\left\langle \prod\limits_{{e_i} > t} {{m_i}\left( x \right)}\right\rangle $. 
\end{lemma}

\section{Constructions of MDS and AMDS Symbol-Pair Codes}
In this section, $ p $ is an odd prime and ${{\Bbb F}_p}$ is a $p$-ary finite field. All symbol-pair codes are constructed by repeated-root cyclic codes.

For MDS symbol-pair codes, we construct  more general generator polynomials with code length $lp$. We also derive all MDS symbol-pair codes of length  $3p$, when $\deg(g(x))\le10$.
 
For AMDS symbol-pair codes, we propose two new classes of AMDS symbol-pair codes from repeated-root cyclic codes by analyzing the solutions of certain equations over ${{\Bbb F}_p}$. We also obtain all AMDS symbol-pair codes of length  $3p$, when $\deg(g(x))\le 10$.
\subsection[Constructions of MDS Symbol-Pair Codes]{MDS Symbol-Pair Codes of length $lp$}
In this subsection, we prove that there exists more general generator polynomials about MDS $ {(lp,6)}_p $ and MDS $ {(lp,5)}_p $ symbol-pair codes. 

For preparation, we define the following notation. Let ${m_1}$ and ${m_2}$ be the element of  ${\Bbb F}_p$, then $\lcm [{m_1},{m_2}]$ represents the lowest common multiple of ${m_1}$ and ${m_2}$ and $\gcd({m_1},{m_2})$ represents the greatest common divisor of ${m_1}$ and ${m_2}$.

Let $\mathcal{C}$ be the cyclic codes in ${{\Bbb F}_p}\left[ {\text{\textit{x}}} \right]/\left\langle {{x^n} - 1} \right\rangle$ and the generator ploynomial of $ \mathcal{C} $  is
$$g\left( x \right) = {\left( {x - 1} \right)^{r_1}}\left( {{x} + 1} \right)^{r_2}\left( {{x} - \omega} \right)^{r_3},$$
where $\omega$ is a primitive $ l$-th  root of unity in ${{\Bbb F}_p}$ and ${r_1}+{r_2}+{r_3}=4$. 
Dinh \cite{Dinh-2018-TIT} \cite{Dinh-2020-TIT} discussed all cases of $l=1$ and $l=2$, here we focus on the case of $l > 2$.

\begin{theorem} \label{theorem 3.1}
	$\mathcal{C}$ is an MDS symbol-pair code with minmum symbol-pair distance $d_p=6$, if ${r_1}$, ${r_2}$ and ${r_3}$ meet the conditions in Table \ref{table-1}  .
\end{theorem}

\begin{table}[h]\label{table1}
	\begin{center}
		\begin{minipage}{\textwidth}
			\caption{MDS symbol-pair codes of Theorem \ref{theorem 3.1}}\label{table-1}
			\begin{center}
				\begin{tabular}{@{}lllll@{}}			
					\toprule
					
						$r_1$ & $r_2$ & $r_3$ & \makecell{$(n,d_{p})_{p}$} & Ref. \\
					\midrule	
					4 & 0 & 0 & \makecell{$(p,6)_{p}$ } &Reference \cite{Dinh-2018-TIT}\\
					\midrule 
					3 & 0 & 1 & \makecell{$(lp,6)_{p}$} & Reference \cite{Zhao-2019-Doc} \\
					\midrule
					2 & 1 & 1 & \makecell{$(klp,6)_{p}$ }&Proposition \ref{Proposition 3.2}\\ 					
					\midrule
					2 & 0 & 2 & \makecell{$(lp,6)_{p}$ } &Proposition \ref{Proposition 3.4}\\
				
				\footnotetext {When $l$ even, $(klp,6)_{p} $ means that $(klp,6)_{p}=(lp,6)_{p}$.}
				\footnotetext {When $l$ odd, $(klp,6)_{p} $ means that $(klp,6)_{p}=(2lp,6)_{p}$.}
					\end{tabular}
		\end{center}
	\end{minipage}
\end{center}
\end{table}

The proof of Theorem \ref{theorem 3.1} needs the following three propositions. For the case of generator polynomials with three factors ${\left( {x - 1} \right)},\left( {{x} + 1} \right)$ and $\left( {{x} - \omega} \right)$, we have the following proposition.

\begin{Proposition}\label{Proposition 3.2}
$ \mathcal{C} $ is an MDS symbol-pair code with $d_p=6$, if one of $\left( {r_1},{r_2},{r_3}\right) =\left(2,1,1 \right) $, $\left(1,2,1 \right)$ and $\left(1,1,2 \right) $ is satisfied.

\end{Proposition}
\begin{proof}
When	$\left( {r_1},{r_2},{r_3}\right) =\left(2,1,1 \right) $, let $\mathcal{C}$ is the cyclic code in ${{\Bbb F}_p}\left[ {\text{\textit{x}}} \right]/\left\langle {{x^n} - 1} \right\rangle$ and generated by $$g\left( x \right) = {\left( {x - 1} \right)^2}\left( {{x} + 1} \right)\left( {{x} - \omega} \right).$$ 

\begin{list}{}{}
	\item 
	By   Lemma \ref{lemma 3.1} , let ${\overline g}_{t}(x) $ be the generator ploynomials of $ \mathcal{\overline C}_{t} $.
\end{list}

\begin{itemize}
	
	\item \noindent	If $ t=0 $, then we have $$ {\overline g}_{0}(x)=(x-1)(x+1)(x-\omega) .$$ 
	It is easy to verify that the minmum Hamming diatance is $3$ in $ \mathcal{\overline C}_{0} $ and $ P_{0}=1 $. Therefore, this indicates $ P_{0} \cdot d_{H}(\mathcal{\overline C}_{0})=3.$
	
	\item \noindent If $ t=1 $, then we have ${\overline g}_{1}(x)=(x-1) $ and $ P_{1}=2 $. Thus, one can derive that ${P_1} \cdot {d_H}({{\mathcal{\overline C}}_1}) = 4.$
	
	\item \noindent If $2 \le t \le p-1$, then we have ${\overline g}_{t}(x)=1 $ and $ P_{t} \ge 2 $. This implies that ${P_t} \cdot {d_H}({{\mathcal{\overline C}}_t}) \ge 3.$
	
	Therefore, it can be verified that $\mathcal{C}$ is an $\left[ {lp,\;lp - 4,\;3} \right]$ repeated-root cyclic code over ${{\Bbb F}_p}$. Lemma \ref{lemma 3.3 } yields that ${d_p} \geqslant 5$, since $\mathcal{C}$ is not an MDS cyclic code.
	
\end{itemize}

\begin{enumerate}
	\item If $c \in \mathcal{ C}$ has $ \omega_p =5$ with  $\omega_H =4 $,
	then its certain cyclic shift must have the form
	$$ \left( { \star ,\; \star ,\; \star ,\; \star ,{0_s}} \right),$$
	where each $ \star $ denotes an element in  ${\Bbb F}_p^{\text{*}}$ and ${0_s}$ is all-zero vector. Without loss of generality, suppose that the first coordinate of $c\left( x \right)$ is 1. We denote that
	$$c\left( x \right) = 1 + {a_1}x + {a_2}{x^2} + {a_3}{x^3} ,$$
	then this leads to $\deg\left( {c\left( x \right)} \right) =3 < 4 = \deg\left( {g\left( x \right)} \right)$.
	
	\item If $c \in \mathcal{ C}$ has  $ \omega_p =5$ with $\omega_H =3 $, then its certain cyclic shift must have the form
	$$\left( { \star ,\; \star ,\;{0_{s_1}},\; \star ,\;{0_{s_2}}} \right),$$
	where each $ \star $ denotes an element in ${\Bbb F}_p^{\text{*}}$ and ${0_{s_1}}$, ${0_{s_2}}$ are all-zero vectors. Without loss of generality, suppose that the first coordinate of $c\left( x \right)$ is 1. We denote that
	$$c\left( x \right) = 1 + {a_1}x + {a_2}{x^t}.$$
	
	When $ t $ even, it can be verified that
	\begin{equation*}  
		\left\{  
		\begin{array}{lr}  
			1 + {a_1} + {a_2}  = 0, &  \\  
			1 - {a_1} + {a_2}= 0, &  
		\end{array}  
		\right.  
	\end{equation*}
	since $c\left( 1 \right) = c\left( { - 1} \right)  = 0$. Then one can derive that $ 2{a_1}=0 $, which is impossible, since $ {a_1} \ne 0 $ and $2\ne 0$.
	
	Similarly, if $ t $ odd, one can obtain that $ 2=0 $, which contradicts $ p $ odd.

\end{enumerate}

	Let $y = -x,z = \frac{x}{\omega }$, For the generator polynomial $g_{r_{1}r_{2}r_{3}}(x)$ of ${\mathcal{C}_{r_{1}r_{2}r_{3}}}$, we can deduce the following results by deforming it,
\[\begin{gathered}
	{g_{211}}(x) = (x - 1)^{2}(x + 1)(x - {\omega}) \hfill \\
=(y + 1)^{2}(y - 1)(y + {\omega})= {g_{121}}(y) \hfill \\
= \omega ^4 {(z - \frac{1}{{{\omega }}})^2}{(z + \frac{\omega }{{{\omega }}})}{(z - \frac{{{\omega }}}{{{\omega }}})}=\omega^4{g_{112}}(z).\hfill \\
\end{gathered} \]

	When $\left( {r_1},{r_2},{r_3}\right) =\left(1,2,1 \right)$ and $\left(1,1,2 \right) $, it is equivalent to  $\left( {r_1},{r_2},{r_3}\right) =\left(2,1,1 \right)  $.

\end{proof}

We have the following two propositions for generator polynomials only two of these three factors  ${\left( {x - 1} \right)},\left( {{x} + 1} \right)$ and $\left( {{x} - \omega} \right)$.

\begin{Proposition} \label{Proposition 3.3}

$ \mathcal{C} $ is an MDS symbol-pair code with $d_p=6$, if one of $\left( {r_1},{r_2},{r_3}\right) =\left(1,0,3 \right) $, $\left(0,3,1 \right)$ and $\left(0,1,3 \right) $ is satisfied.	
\end{Proposition}
\begin{proof}
When 	$\left( {r_1},{r_2},{r_3}\right) =\left(1,0,3 \right) $, let $\mathcal{C}$ be a cyclic code in ${{\Bbb F}_p}\left[ {\text{\textit{x}}} \right]/\left\langle {{x^n} - 1} \right\rangle$ generated by $$g\left( x \right) = {\left( {x - 1} \right)}\left( {{x} - \omega} \right)^3.$$ 

\begin{description}
	\item[] 
	\noindent By Lemma \ref{lemma 3.1}, minmum Hamming distance  $d_H=3$ of $ \mathcal{C} $ can be derived and Lemma \ref{lemma 3.3 } implies that $d_p\ge 5$.
	
	With arguments similar to the previous Proposition \ref{Proposition 3.2}, there are no codewords of $\mathcal{C}$ with Hamming weight  $\omega_H=4$ such that the 4 nonzero terms appear with consecutive coordinates. 
	
	We are left to show that there are no codewords of $\mathcal{C}$ with Hamming weight $\omega_H=3$ in the form
	$$\left( { \star ,\; \star ,\;{0_{s_1}},\; \star ,\;{0_{s_2}}} \right),$$

	where each $ \star $ denotes an element in ${\Bbb F}_p^{\text{*}}$ and ${0_{s_1}}$, ${0_{s_2}}$ are all-zero vectors. Without loss of generality, suppose that the first coordinate of $c\left( x \right)$ is 1. We denote that
	$$c\left( x \right) = 1 + {a_1}x + {a_2}{x^t}.$$
	Then $c'\left( \omega \right) = c''\left( \omega \right) =0$ induces that $t-1=kp$ for some positive integers $k\le l-2$, together with $c\left( { \omega} \right) =0$, one can immediately get
	\begin{equation*}  
		\left\{  
		\begin{array}{lr}  
			1 + {a_1}\omega + {a_2}\omega^{k+1}  = 0, &  \\  
			{a_1} + {a_2}\omega^k= 0. &  
		\end{array}  
		\right.  
	\end{equation*}
	By solving the system, we have $ 1=0 $, which derive a contradiction.

\end{description}

For the case of $\left( {r_1},{r_2},{r_3}\right) =\left(0,3,1 \right),\left(0,1,3 \right) $, with similar to the previous Proposition \ref{Proposition 3.2}, we can prove that these repeated-root cyclic codes are equivalent.

Therefore, $\mathcal{C}$ is an MDS symbol-pair code with  minimum symbol-pair weight $d_p= 6$, when $\left( {r_1},{r_2},{r_3}\right) =\left(1,0,3 \right) ,\left(0,3,1 \right)$ and $\left(0,1,3 \right) $.
\end{proof}

\begin{Proposition} \label{Proposition 3.4}
$ \mathcal{C} $ is an MDS symbol-pair code with $d_p=6$, if one of $\left( {r_1},{r_2},{r_3}\right) =\left(2,0,2 \right)$ and $\left(0,2,2 \right) $ is satisfied.		
\end{Proposition}
\begin{proof}

When 	$\left( {r_1},{r_2},{r_3}\right) =\left(2,0,2 \right) $, let $\mathcal{C}$ be a cyclic code in ${{\Bbb F}_p}\left[ {\text{\textit{x}}} \right]/\left\langle {{x^n} - 1} \right\rangle$ generated by $$g\left( x \right) = {\left( {x - 1} \right)^2}\left( {{x} - \omega} \right)^2.$$ 

\begin{description}
	\item[] 
	\noindent By Lemma \ref{lemma 3.1}, minmum Hamming distance  $d_H=3$  of $ \mathcal{C} $ can be derived and Lemma \ref{lemma 3.3 } implies that $d_p\ge 5$.
	
	With arguments similar to the previous Proposition \ref{Proposition 3.2}, there are no codewords of $\mathcal{C}$ with Hamming weight  $\omega_H=4$ such that the 4 nonzero terms appear with consecutive coordinates. 
	
	We are left to show that there are no codewords of $\mathcal{C}$ with Hamming weight $\omega_H=3$ in the form
	$$ \left( { \star ,\; \star ,\;{0_{s_1}},\; \star ,\;{0_{s_2}}} \right),$$
	where each $ \star $ denotes an element in ${\Bbb F}_p^{\text{*}}$ and ${0_{s_1}}$, ${0_{s_2}}$ are all-zero vectors. Without loss of generality, suppose that the first coordinate of $c\left( x \right)$ is 1. We denote that
	$$c\left( x \right) = 1 + {a_1}x + {a_2}{x^t}.$$
	However, $c'\left( 1 \right)= c'\left( \omega \right) = 0$ can deduce that $l$ is a divisor of $t-1$ , combine with $c\left( 1 \right)= c\left( \omega \right) = 0$, we have 
	\begin{equation*}  
		\left\{  
		\begin{array}{lr}  
			1 + {a_1} + {a_2}  = 0, &  \\  
			1 + {a_1}\omega + {a_2}\omega= 0, &  
		\end{array}  
		\right.  
	\end{equation*}	
	then we have $\omega=1$, which is impossible, since $\omega^l=1$ and $l> 2$.

\end{description}

When $\left( {r_1},{r_2},{r_3}\right) =\left(0,2,2 \right)$, 
with similar to the previous Proposition \ref{Proposition 3.2}, we can prove that these repeated-root cyclic codes are equivalent.

\end{proof}

In fact, Lemma \ref{lemma 3.8 } is a special form of Proposition \ref{Proposition 3.3} and Proposition \ref{Proposition 3.4} for $d_{p}=6$, where $\omega$ is a $p-1\over 2$-th primitive element in ${{\Bbb F}_p}$.

\begin{remark} \label{remark 3.5}
	When $l$ even, factor ${\left( {x + 1} \right)}$ is contained in $x^{l}-1$;
	when $l$ odd, factor ${\left( {x + 1} \right)}$ is contained in $x^{2l}-1$.
\end{remark}

From Lemma \ref{lemma 3.7 } to Lemma \ref{lemma 3.9 } and Proposition \ref{Proposition 3.2} to Proposition \ref{Proposition 3.4}, we find all MDS symbol-pair codes containing these three factors ${\left( {x - 1} \right)},\left( {{x} + 1} \right)$ and $\left( {{x} - \omega} \right)$ with minmum symbol-pair distance $d_p=6$.

In what follows, we obtain more general generator polynomials for
 symbol-pair codes with length $n = lp$  and minimum symbol-pair distance 5 or 6.

Let $\mathcal{C}_{a}$ be the cyclic codes in ${{\Bbb F}_p}\left[ {\text{\textit{x}}} \right]/\left\langle {{x^n} - 1} \right\rangle$ and the generator ploynomial of $ \mathcal{C}_{a} $ is
$$g_{a}\left( x \right) = \left( {{x} - \omega^{t_1}} \right)^{r_1}\left( {{x} - \omega^{t_2}} \right)^{r_2},$$ 
where ${t_1}\ge{t_2},$ $\ord(\omega^{t_1})={m_1}$, $\ord(\omega^{t_2})={m_2}$, $\lcm[{m_1},{m_2}]=l,\;\gcd\left({t_1}-{t_2},l \right) =1,\;3\le {r_1}+{r_2}\le 4$ and $\omega$ is primitive element in ${{\Bbb F}_p}$.

\begin{corollary}\label{theorem 3.6}
$\mathcal{C}_{a}$ is an MDS symbol-pair code,  if  ${r_1}\ne 0$ and ${r_2}\ne 0$.
\end{corollary}
\begin{proof}
There are three cases that need to be discussed, $({r_1},{r_2} ) =\left(2,1 \right) $,  $\left(1,3 \right) $ and $\left(2,2 \right)$. When $({r_1},{r_2} ) =\left(1,2 \right) $ and  $\left(3,1 \right) $ is satisfied, it is similar to $({r_1},{r_2} ) =\left(2,1 \right) $ and  $\left(1,3 \right) $. 
\begin{description}

	\item [\textbf{Case I.}] For the case of $({r_1},{r_2} ) =\left(2,1 \right) $,
	if there exsits a nonzero codeword with minmum Hamming weight $\omega_H=2$ in $\mathcal{C}_{a}$, 
	without loss of generality, suppose that the first coordinate of $c\left( x \right)$ is 1. We denote that
	$$c(x)=1+a_1x^t,$$ 
	where $a_1\ne 0$ and $t\ne 0$.	It follows from $c(\omega^{t_1})=c(\omega^{t_2})=0$ that 
	\begin{equation*}  
		\left\{  
		\begin{array}{lr}  
			1+{a_1} \omega^{tt_1} = 0, &  \\  
			1+{a_1} \omega^{tt_2}= 0. &    
		\end{array}  
		\right.
	\end{equation*}
	By solving the system, we have	${\omega ^{t\left( {{t_1} - {t_2}} \right)}} = 1$. Together with   $\gcd\left({t_1}-{t_2},l \right) =1$ and $c'(\omega^{t_1})=0$, one can immediately get $lp$ is a divisor of $t$, which contradicts with the code length $lp$. 
	
	Thus, combined with Lemma \ref{lemma 3.1}, the minmum Hamming distance of $\mathcal{C}_{a}$ is $d_H=3$. 
	By Lemma \ref{lemma 3.3 }, we have $\mathcal{C}_{a}$ is an MDS $ (lp,5)_{p} $ symbol-pair code.
	
	\item	[\textbf{Case II.}] For the case of	 $({t_1},{t_2} ) =\left(3,1 \right) $, we have the generator polynomial	
	$$g_{a}\left( x \right) = \left( {{x} - \omega^{t_1}} \right)^3\left( {{x} - \omega^{t_2}} \right).$$
	By the proof of $   {\bf{Case\; I}}$, since $ {\bf{Case\; II}}$ is a subcode of $ {\bf{Case\; I}}$, we can draw the conclusion that the minmum Hamming weight of $\mathcal{C}_{a}$ is $3$ in $  {\bf{Case\; II}}$, when  $({t_1},{t_2} ) =\left(3,1 \right) $. 
	
	If there  is a codeword with Hamming weight $3$ and symbol-pair weight $5$. Then its certain cyclic shift must be the following form
	$$ \left( { \star ,\; \star ,\;{0_{s_1}},\; \star ,\;{0_{s_2}}} \right),$$ where each $ \star $ denotes an element in ${\Bbb F}_p^{\text{*}}$ and ${0_{s_1}}$, ${0_{s_2}}$ are all-zero vectors. Then we have a codeword polynomial
	$$c\left( x \right) = 1 + {a_1}x + {a_2}{x^t}.$$
	However, it follows from $c'\left( \omega^{t_1} \right)= c''\left( \omega^{t_1} \right) =0$ that 
	\begin{equation*}  
		\left\{  
		\begin{array}{lr}  
			{a_1} + t{a_2}\omega^{(t-1)t_1} = 0, &\\
			t(t-1){a_2}\omega^{(t-2)t_1} = 0. &    
		\end{array}  
		\right.  
	\end{equation*}
	By solving the system, we have ${p\,\left| \,{t - 1} \right.}$.	Then $c'\left( \omega^{t_1} \right)= c\left( \omega^{t_1} \right) =0$  indicates
	\begin{equation*}  
		\left\{  
		\begin{array}{lr}  
			1+	{a_1}\omega^{t_1} + {a_2}\omega^{{t_1}t} = 0, & \\ 
			{a_1}  + {a_2}\omega^{{t_1}(t-1)} = 0, &   
		\end{array}  
		\right.  
	\end{equation*}	
	which means that $1=0$, a contradiction.	
	
	Therefore, there no exsits a nonzero codeword with Hamming weight $\omega_H=3$ and symbol-pair weight $\omega_p=5$ and $\mathcal{C}_{a}$ is an MDS $ (lp,6)_{p} $ symbol-pair code, when $({r_1},{r_2} ) =\left(3,1 \right) $.
	
	\item [\textbf{Case III.}] For the case of	 $({t_1},{t_2} ) =\left(2,2 \right) $, similarly, we have the generator polynomial	
	$$g_{a}\left( x \right) = \left( {{x} - \omega^{t_1}} \right)^2\left( {{x} - \omega^{t_2}} \right)^2,$$
	by the proof of $   {\bf{Case\; I}}$, we can draw the conclusion that the minmum Hamming weight of $\mathcal{C}_{a}$ is $3$. 
	
	Similar to $  {\bf{Case\; II}}$, there is no codeword with Hamming weight of $\omega_H=4$ and symbol-pair weight of $\omega_p=5$. 
	
	If there exsits a codeword with Hamming weight $\omega_H=3$ and symbol-pair weight $\omega_p=5$, the codeword certain cyclic shift must have a form
	$$\left( { \star ,\; \star ,\; \star ,\; \star ,{0_s}} \right),$$
	where each $ \star $ denotes an element in  ${\Bbb F}_p^{\text{*}}$ and ${0_s}$ is all-zero vector. Without loss of generality, suppose that the first coordinate of $c\left( x \right)$ is 1. We denote that
	$$c\left( x \right) = 1 + {a_1}x + {a_2}{x^t}.$$
	However, $c'\left( \omega^{t_1} \right)= c'\left( \omega^{t_2} \right) =0$ induces that
	\begin{equation*}  
		\left\{  
		\begin{array}{lr}  
			{a_1} + t{a_2}\omega^{(t-1)t_1} = 0, &\\
			{a_1} + t{a_2}\omega^{(t-1)t_2} = 0. &    
		\end{array}  
		\right.  
	\end{equation*}
	By solving the system, we have $\omega^{(t-1)(t_1-t_2)}=1$, since $t\ne kp$, otherwise ${a_1} = 0$. Together with $\gcd\left({t_1}-{t_2},l \right) =1$, one can immediately get ${l\,\left| \,{t - 1} \right.}$ and ${a_1} + t{a_2} = 0$. Combined with $c\left( \omega^{t_1} \right)= c\left( \omega^{t_1} \right) =0$, we have
	\begin{equation*}  
		\left\{  
		\begin{array}{lr}  
			1+	{a_1}\omega^{t_1} + {a_2}\omega^{t_1} = 0, & \\ 
			1+	{a_1}\omega^{t_2} + {a_2}\omega^{t_2} = 0, &   
		\end{array}  
		\right.  
	\end{equation*}
	which implies ${a_1} + {a_2} = 0$. Thus, we have ${p\left| {t - 1} \right.}$, which contradicts the code length $lp$.

\end{description}

Therefore, we prove that there no exsits a nonzero codeword with Hamming weight $\omega_H=3$ and symbol-pair weight $\omega_p=5$ and $\mathcal{C}_{a}$ is an MDS $ (lp,6)_{p} $ symbol-pair code, when $({r_1},{r_2} ) =\left(2,2 \right) $ and  $\gcd\left({t_1}-{t_2},l \right) =1$.

\end{proof}	

In fact, similar to previous Proposition \ref{Proposition 3.2}, \textbf{Case I} is equivalent to $g (x) = (x-1) ^ 2 (x - \omega)$ in Chen \cite{Chen-2017-CIT}, \textbf{Case II} and \textbf{Case III} are equivalent to Proposition \ref{Proposition 3.3} and Proposition \ref{Proposition 3.2}, respectively.

We use an example to illustrate that the repeated-root cyclic codes of the generator polynomials with the same forms in the above  Corollary \ref{theorem 3.6} are not all MDS symbol-pair codes.
\begin{example} \label{example 3.6}
Let $\mathcal{C}$ and  be a cyclic code in ${{\Bbb F}_5}\left[ {\text{\textit{x}}} \right]/\left\langle {{x^{20}} - 1} \right\rangle$ and the generator ploynomial of $ \mathcal{C} $  is
$$g\left( x \right) = {\left( {x - 2} \right)^2}\left( {{x} - 3} \right),$$
where $\omega=3$ is a primitive element in ${{\Bbb F}_5}$ and $2=3^3$. Then we have the minmum Hamming distance $d_H=2$ by a magma progarm. Therefore, $\mathcal{C}$ is not an MDS symbol-pair code. 

Similarly, when the generator ploynomial of $ \mathcal{C} $  is one of 	$$g\left( x \right) = {\left( {x - 2} \right)^3}\left( {{x} - 3} \right),$$
$$g\left( x \right) = {\left( {x - 2} \right)^2}\left( {{x} - 3} \right)^2,$$
$ \mathcal{C} $ is still not an MDS symbol-pair code, since minmum Hamming distance is $d_H=2$.
\end{example}




\subsection[Constructions of MDS Symbol-Pair Code]{MDS Symbol-Pair Codes of length $3p$}
In this subsection, for $n = 3p$, we obtain all MDS symbol-pair codes of $d_{p}\le 12$ from repeated-root cyclic codes and all AMDS symbol-pair codes of $d_{p}< 12$ from repeated-root cyclic codes. Furthermore, we discuss all minimum symbol-pair distance of the repeated-root cyclic codes with code length of $3p$, when $n-k < 10$.  For preparation, we define the following notation.

Let $\mathcal{C}_{r_{1}r_{2}r_{3}}$ be the cyclic code in ${{\Bbb F}_p}\left[ {\text{\textit{x}}} \right]/\left\langle {{x^n} - 1} \right\rangle$ and the generator ploynomial of $ \mathcal{C}_{r_{1}r_{2}r_{3}} $ is 
$$g_{r_{1}r_{2}r_{3}}\left( x \right) = (x - 1)^{r_{1}}(x - \omega)^{r_{2}}(x - {\omega ^2})^{r_{3}}.$$
where $\omega$ is a primitive $ 3$-th  root of unity in ${{\Bbb F}_p}$ and $r_{i}\le p-1,i=1,2,3$. 
\begin{Proposition}\label{proposion 3.9}
	Repeated-root cyclic codes  $\mathcal{C}_{r_{1}r_{2}r_{3}}$ are equivalent, where  the exponents of the three factors of the generator polynomial can be exchanged with each other.
\end{Proposition}
\begin{proof}
	We first prove that such repeated-root cyclic codewords $${\mathcal{C}} = \left\langle {{{(x - \omega^i)}^{r_{1}}}{{(x - \omega^{i+1})}^{r_{2}}}{{(x - {\omega ^{i+2}})}^{r_{3}}}} \right\rangle $$ are equivalent for $i=0,1,2$. 
	
	Without loss of generality, suppose that $${\mathcal{C}_{r_{1}r_{2}r_{3}}} = \left\langle {{{(x - 1)}^{r_{1}}}{{(x - \omega)}^{r_{2}}}{{(x - {\omega ^2})}^{r_{3}}}} \right\rangle ,$$ 
	 $${\mathcal{C}_{r_{3}r_{1}r_{2}}} = \left\langle {{{(x - \omega)}^{r_{1}}}{{(x - \omega^2)}^{r_{2}}}{{(x - 1)}^{r_{3}}}} \right\rangle $$ 
	 and 
	 $${\mathcal{C}_{r_{2}r_{3}r_{1}}} = \left\langle {{{(x - {\omega ^2})}^{r_{1}}}{{(x - 1)}^{r_{2}}}{{(x - \omega)}^{r_{3}}}} \right\rangle .$$ 
	 We denote that $g_{r_{1}r_{2}r_{3}}(x)$, $g_{r_{3}r_{1}r_{2}}(x)$ and $g_{r_{2}r_{3}r_{1}}(x)$ represent the generator polynomials of ${\mathcal{C}_{r_{1}r_{2}r_{3}}}$, ${\mathcal{C}_{r_{3}r_{1}r_{2}}}$ and ${\mathcal{C}_{r_{2}r_{3}r_{1}}}$, respectively.
	
	Let $y = \frac{x}{{{\omega ^2}}},z = \frac{x}{\omega }$, For the generating polynomial $g_{r_{1}r_{2}r_{3}}(x)$ of ${\mathcal{C}_{r_{1}r_{2}r_{3}}}$, we can deduce the following results by deforming it.
	\[\begin{gathered}
		{g_{r_{1}r_{2}r_{3}}}(x) = {(x - 1)^{{r_1}}}{(x - \omega )^{{r_2}}}{(x - {\omega ^2})^{{r_3}}} \hfill \\
		= {\omega ^{2\left( {{r_1} + {r_2} + {r_3}} \right)}}{(\frac{x}{{{\omega ^2}}} - \frac{1}{{{\omega ^2}}})^{{r_1}}}{(\frac{x}{{{\omega ^2}}} - \frac{\omega }{{{\omega ^2}}})^{{r_2}}}{(\frac{x}{{{\omega ^2}}} - \frac{{{\omega ^2}}}{{{\omega ^2}}})^{{r_3}}} \hfill \\
		= {\omega ^{2\left( {{r_1} + {r_2} + {r_3}} \right)}}{(y - \omega )^{{r_1}}}{(y - {\omega ^2})^{{r_2}}}{(y - 1)^{{r_3}}} \hfill \\
		={\omega ^{2\left( {{r_1} + {r_2} + {r_3}} \right)}}{g_{r_{3}r_{1}r_{2}}}(y) \hfill \\
		= {\omega ^{{r_1} + {r_2} + {r_3}}}{(\frac{x}{\omega } - \frac{1}{\omega })^{{r_1}}}{(\frac{x}{{{\omega ^2}}} - \frac{\omega }{\omega })^{{r_2}}}{(\frac{x}{\omega } - \frac{{{\omega ^2}}}{\omega })^{{r_3}}} \hfill \\
		= {\omega ^{{r_1} + {r_2} + {r_3}}}{(z - {\omega ^2})^{{r_1}}}{(z - 1)^{{r_2}}}{(z - \omega )^{{r_3}}} \hfill \\ 
		= {\omega ^{{r_1} + {r_2} + {r_3}}}{g_{r_{2}r_{3}r_{1}}}(z). \hfill \\
	\end{gathered} \]
	\noindent Thus, repeated-root cyclic codes $${\mathcal{C}} = \left\langle {{{(x - \omega^i)}^{r_{1}}}{{(x - \omega^{i+1})}^{r_{2}}}{{(x - {\omega ^{i+2}})}^{r_{3}}}} \right\rangle $$ are equivalent for $i=0,1,2$.
	
	Next, since both $\omega$ and $\omega^2$ are primitive 3-th root of unity in ${\Bbb F}_{p}$, we have repeated-root cyclic code  $${\mathcal{C}_{r_{1}r_{2}r_{3}}} = \left\langle {{{(x - 1)}^{r_{1}}}{{(x - \omega)}^{r_{2}}}{{(x - {\omega ^2})}^{r_{3}}}} \right\rangle $$ and repeated-root cyclic code  $${\mathcal{C}_{r_{1}r_{3}r_{2}}} = \left\langle {{{(x - 1)}^{r_{1}}}{{(x - {\omega ^2})}^{r_{2}}}}{{(x - \omega)}^{r_{3}}} \right\rangle ,$$ which are equivalent.
	
	In conclusion, we prove all cases of this proposition, i.e. repeated-root cyclic codes  ${\mathcal{C}_{r_{1}r_{2}r_{3}}}$, ${\mathcal{C}_{r_{1}r_{3}r_{2}}}$,
	${\mathcal{C}_{r_{2}r_{1}r_{3}}}$, ${\mathcal{C}_{r_{2}r_{3}r_{1}}}$, ${\mathcal{C}_{r_{3}r_{1}r_{2}}}$ and ${\mathcal{C}_{r_{3}r_{2}r_{1}}}$ are equivalent to each other.
\end{proof}
The above Proposition \ref{proposion 3.9} shows that the exponential positions of the three factors $x-1$, $x-\omega$ and $x-\omega^2$ of the generator polynomial of 
$ \mathcal{C}_{r_{1}r_{2}r_{3}} $ are equivalent. Without loss of generality, suppose that $p-1 \ge r_{1}\ge r_{2}\ge r_{3}\ge 0$ in the next part of this subsection. Then we have the following theorem.
\begin{theorem}\label{theorem of 3p MDS}
	$ \mathcal{C}_{r_{1}r_{2}r_{3}} $ is an MDS symbol-pair codes over ${\Bbb F}_{p}$, if one of the following two conditions is true
	\begin{enumerate}
		\item $r_{1}\le 5$, $0\le r_{2}-r_{3}\le 1$ and $ {r_{1}} = r_{2}+r_{3}$,
		\item $r_{1}< 5$, $0\le r_{2}-r_{3}\le 1$ and $ {r_{1}} = r_{2}+r_{3}+1$.
	\end{enumerate}
\end{theorem}
Based on previous work and the results of this paper, all known MDS symbol-pair codes with $n-k\le 10$ from repeated-root cyclic codes  ${\mathcal{C}_{r_{1}r_{2}r_{3}}}, $ which are listed in the following Table \ref{table-4}.
\begin{table}[h]\label{table4}
	\begin{center}
		\begin{minipage}{\textwidth}
			\caption{All MDS symbol-pair codes of length $3p$ for $d_{p}\le 12$}\label{table-4}
			\begin{center}
				\begin{tabular}{@{}lllll@{}}			
					\toprule
					
					$r_1$ & $r_2$ & $r_3$ & \makecell{$(n-k,d_{p})_{p}$} & Ref. \\
					\midrule	
					0 & 2 & 0 & \makecell{$(2,4)_{p}$} & Trivially($r_1= 0$)\\
					\midrule  
					2 & 1 & 0 & \makecell{$(3,5)_{p}$} & Reference \cite{Chen-2017-CIT} \\
					\midrule			
					3 & 1 & 0 & \makecell{$(4,6)_{p}$} & Reference \cite{Chen-2017-CIT} \\ 
					\midrule					
					3 & 1 & 1 & \makecell{$(5,7)_{p}$} & Reference \cite{Chen-2017-CIT}\\ 			
					\midrule
					3 & 2 & 1 & \makecell{$(6,8)_{p}$} &Reference \cite{Chen-2017-CIT}\\
					\midrule
					4 & 2 & 2 & \makecell{$(8,10)_{p}$}&Reference \cite{Ma-2022-DCC} \\ 
					\midrule
					5 & 3 & 2 & \makecell{$(10,12)_{p}$} &Reference \cite{Ma-2022-DCC} \\
					\midrule
					2 & 1 & 1 & \makecell{$(4,6)_{p}$} & Theorem \ref{theorem 3.1} \\
					\midrule
					2 & 2 & 0 & \makecell{$(4,6)_{p}$} & Theorem \ref{theorem 3.1} \\
					\midrule
					4 & 2 & 1 & \makecell{$(7,9)_{p}$} &Proposition \ref{Proposition C_{421}}\\
					\footnotetext{These MDS symbol-pair codes are constructed by repeated-root cyclic codes.}	
				\end{tabular}
			\end{center}
		\end{minipage}
	\end{center}
\end{table}

Similar to the previous MDS symbol-pair codes from repeated-root cyclic codes  ${\mathcal{C}_{r_{1}r_{2}r_{3}}}, $
for  AMDS symbol-pair codes with $n-k\le 10$, which are listed in the following Table \ref{table-5}.

\begin{table}[h]\label{table5}
	\begin{center}
		\begin{minipage}{\textwidth}
			\caption{All AMDS symbol-pair codes of length $3p$ for $d_{p}< 12$}\label{table-5}
			\begin{center}
				\begin{tabular}{@{}lllll@{}}			
					\toprule
					
					$r_1$ & $r_2$ & $r_3$ & \makecell{$(n-k,d_{p})_{p}$} & Ref. \\
					\midrule
					4 & 3 & 2 & \makecell{$(9,10)_{p}$} &Reference \cite{Ma-2022-DCC} \\
					\midrule 
					0 & 3 & 0 & \makecell{$(3,4)_{p}$} & Proposition \ref{proposition of C_{rrr}}\\
					\midrule
					2 & 2 & 1 & \makecell{$(5,6)_{p}$} & Proposition \ref{proposition of C_{rrr}}\\ 
					\midrule					
					3 & 2 & 0 & \makecell{$(5,6)_{p}$} &Proposition \ref{proposition of C_{rrr}}\\ 
					\midrule										
					3 & 2 & 2 & \makecell{$(7,8)_{p}$} &Proposition \ref{proposition of C_{3rr}} \\ 
					\midrule
					3 & 3 & 1 & \makecell{$(7,8)_{p}$} &Proposition \ref{proposition of C_{3rr}} \\ 
					\midrule
					4 & 1 & 0 & \makecell{$(5,6)_{p}$} & Proposition \ref{proposition of C_{rrr}} \\
					\midrule 
					4 & 1 & 1 & \makecell{$(6,7)_{p}$} &Proposition \ref{proposition of C_{r11}} \\
					\midrule
					4 & 3 & 1 & \makecell{$(8,9)_{p}$} &Proposition \ref{proposition of C_{rr1}}\\
					\midrule
					5 & 2 & 1 & \makecell{$(8,9)_{p}$} &Proposition \ref{proposition of C_{rr1}} \\
					\midrule					
					5 & 2 & 2 & \makecell{$(9,10)_{p}$}&Proposition  \ref{proposition of C_{rr2}} \\ 
					
					\footnotetext{These AMDS symbol-pair codes are constructed by repeated-root cyclic codes.}	
				\end{tabular}
			\end{center}
		\end{minipage}
	\end{center}
\end{table}

Researchers have given some proofs of the Theorem \ref{theorem of 3p MDS}, and the Theorem \ref{theorem 3.1} in the previous paper also includes some proofs. Here we only need to prove that ${\mathcal{C}_{421}} $ is an MDS symbol-pair code.

\begin{Proposition}\label{Proposition C_{421}}
${\mathcal{C}_{421}} $ is an MDS ${\left( {3p,\;9} \right)_p}$ symbol-pair code.	
\end{Proposition} 
\begin{proof}
Since  $\mathcal{C}_{421}= \left\langle {{{(x - 1)}^4}{{(x - \omega )}^2}{{(x - {\omega ^2})}}} \right\rangle ,$ for any codeword $c \in \mathcal{C}$, we have $$c\left( 1 \right) = c\left( \omega  \right) = c\left( {  \omega^2 } \right) = c'\left( 1 \right)= c'\left( \omega \right) = c''\left( 1 \right) =c'''\left( 1 \right) = 0.$$

By Lemma \ref{lemma 3.1}, $\mathcal{C}_{421}$ is a $\left[ {3p,\;3p - 7,\;5} \right]$ cyclic code over ${{\Bbb F}_p}$. Since $\mathcal{C}_{421}$ is a subcode of Lemma \ref{lemma 3.6 }, we have ${d_p} \geqslant 8$.

To prove that $\mathcal{C}_{421}$ is an MDS ${\left( {3p,\;9} \right)_p}$ symbol-pair code, it suffices to verify that there does not exist codeword in $\mathcal{C}_{421}$ with symbol-pair weight 8. Then we have three cases to discuss.
\begin{description}	
	\item[\textbf{Case I.}] If there are codewords with Hamming weight $\omega_H=5$ and symbol-pair weight $\omega_p=8$, then its certain
	cyclic shift must be one of the following forms
	$$ \left( { \star , \star ,{0_{s_1}}, \star ,\star,{0_{s_2}}, \star ,{0_{s_3}}} \right)$$ or
	$$ \left( { \star ,\star , \star ,{0_{s_1}}, \star ,{0_{s_2}}, \star ,{0_{s_3}}} \right),$$
	where each $ \star $ denotes an element in ${\Bbb F}_p^{\text{*}}$ and  ${0_{s_1}}$, ${0_{s_2}}$,${0_{s_3}}$ are all-zero vectors.

	$  {\bf{Subcase\; 1.1}}$ For the case of $$ \left( { \star , \star ,{0_{s_1}}, \star ,\star,{0_{s_2}}, \star ,{0_{s_3}}} \right),$$ without loss of generality, suppose that the first coordinate of $c\left( x \right)$ is 1. We denote that
	
	$$c\left( x \right) = 1 + {a_1}x + {a_2}{x^l} + {a_3}{x^{l + 1}} + {a_4}{x^t}.$$
	
	When $t \equiv 0({\bmod~3}) $ and $l \equiv 0({\bmod~3}) $, it follows from $c\left( 1 \right) = c\left( { \omega} \right) =c\left( { \omega^2} \right) = 0$ that
	\begin{equation*}  
		\left\{  
		\begin{array}{lr}  
			1+{a_1} + {a_2} + {a_3} + {a_4} = 0, &  \\  
			1+{a_1}\omega + {a_2} + {a_3}\omega + {a_4} = 0, &\\  
			1+{a_1}\omega^2 + {a_2} + {a_3}\omega^2 + {a_4}  = 0. &    
		\end{array}  
		\right.  
	\end{equation*}
	By solving the system, we have	 $ a_1=-a_3 $. However, $c'\left( 1 \right) =c'\left( \omega \right) =0$ induces that
	\begin{equation*}  
		\left\{  
		\begin{array}{lr}   
			{a_1} + t{a_2} +\left( t+1\right) {a_3} + l{a_4}  = 0, &\\
			{a_1} + t{a_2}\omega^2 +\left( t+1\right) {a_3} + l{a_4}\omega^2  = 0. &    
		\end{array}  
		\right.  
	\end{equation*}
	Together with $  a_1=-a_3 $, one can immediately get
	
	\begin{equation*}  
		\left\{  
		\begin{array}{lr}  
			t{a_2} + t {a_3} + l{a_4}  = 0, &\\
			t{a_2}\omega^2 +t {a_3} + l{a_4}\omega^2  = 0, &    
		\end{array}  
		\right.  
	\end{equation*}
	and we can get $t(1-\omega^2)a_3=0$, which is impossible, since $t \equiv 0\left( {\bmod~ 3} \right)$ and the code length is $3p$. 
	
	When $l \equiv i({\bmod~3}) $ and $t \equiv j({\bmod~3}),\, i,j=0,1,2 $,
	values in all $i$  and $j$ of ${\bf{Subcase\; 1.1}}$ are shown in the following Table \ref{table-2}.
\end{description}

\begin{table}[h]\label{table2}
	\begin{center}
		\begin{minipage}{\textwidth}
			\caption{Summary of  ${\bf{Subcase\; 1.1}}$}\label{table-2}
			\begin{center}
			\begin{tabular}{@{}lllll@{}}			
				\toprule
			
				$i$ & $j$ & Conditions &\makecell{Results} & Contradictory \\
				\midrule
				0 & 0 &$\left[\kern-0.15em\left[ 1 
				\right]\kern-0.15em\right],\left[\kern-0.15em\left[ 2 
				\right]\kern-0.15em\right]$ & \makecell{$t(1-\omega^2)a_3=0$} & $t\le  3p-2$ \\ 
				\midrule
				0 & 1 &$\left[\kern-0.15em\left[ 1 
				\right]\kern-0.15em\right],\left[\kern-0.15em\left[ 2 
				\right]\kern-0.15em\right],\left[\kern-0.15em\left[ 3 
				\right]\kern-0.15em\right]$& \makecell{$\omega^2-1=0 $  } & $\omega^3=1$ \\ 
				\midrule
				0 & 2 &$\left[\kern-0.15em\left[ 1 
				\right]\kern-0.15em\right],\left[\kern-0.15em\left[ 3 
				\right]\kern-0.15em\right]$ & \makecell{$ a_1+a_3+a_4=0,$\\$a_1+a_3-a_4=0 $ }& $a_4\in {\Bbb F}_p^{\text{*}}$\\
				\midrule
				1 & 0 &$\left[\kern-0.15em\left[ 1 
				\right]\kern-0.15em\right],\left[\kern-0.15em\left[ 3 
				\right]\kern-0.15em\right]$ & \makecell{$ a_1+a_2+a_3=0,$\\$a_1+a_2-a_3=0 $ }& $a_3\in {\Bbb F}_p^{\text{*}}$\\
				\midrule
				1 & 1 & $\left[\kern-0.15em\left[ 1 
				\right]\kern-0.15em\right],\left[\kern-0.15em\left[ 3 
				\right]\kern-0.15em\right]$ & \makecell{$ 3=0 $ }& $p\ne 3$\\ 
				\midrule
				1 & 2 & $\left[\kern-0.15em\left[ 1 
				\right]\kern-0.15em\right],\left[\kern-0.15em\left[ 3 
				\right]\kern-0.15em\right]$ & \makecell{$ 3=0 $ }& $p\ne 3$\\ 
				\midrule
				2 & 0 & $\left[\kern-0.15em\left[ 1 
				\right]\kern-0.15em\right],\left[\kern-0.15em\left[ 3 
				\right]\kern-0.15em\right]$ & \makecell{$ a_1-a_2=0 $,\\$a_1+a_2=0$ }& $a_1,a_2 \in {\Bbb F}_p^{\text{*}}$.\\ 
				\midrule
				2 & 1 & $\left[\kern-0.15em\left[ 1 
				\right]\kern-0.15em\right],\left[\kern-0.15em\left[ 3 
				\right]\kern-0.15em\right]$ & \makecell{$a_1+a_2+a_4=0,$\\$a_1-a_2+a_4=0 $ }& $a_2 \in {\Bbb F}_p^{\text{*}}$\\ 
				\midrule
				2 & 2 &$\left[\kern-0.15em\left[ 1 
				\right]\kern-0.15em\right],\left[\kern-0.15em\left[ 3 
				\right]\kern-0.15em\right]$ &  \makecell{$a_1+a_2+a_4=0,$\\$a_1-a_2-a_4=0 $} &  $a_1 \in {\Bbb F}_p^{\text{*}}$ \\ 

			\end{tabular}
			\footnotetext{ Conditions $\left[\kern-0.15em\left[ 1 \right]\kern-0.15em\right]$, $\left[\kern-0.15em\left[ 2 \right]\kern-0.15em\right]$  and $\left[\kern-0.15em\left[ 3 \right]\kern-0.15em\right]$ represent $c\left( 1 \right) = c\left( { \omega} \right) =c\left( { \omega^2} \right) = 0$, $ c'\left( 1 \right)= c'\left( \omega \right) =0$ and  $c\left( 1 \right) + c\left( { \omega} \right) +c\left( { \omega^2} \right) = 0$, respectively.}

	\end{center}
		\end{minipage}
	\end{center}
\end{table}

\begin{description}
	\item 
	$  {\bf{Subcase\; 1.2}}$ For the subcase of $$ \left( { \star , \star , \star ,{0_{s_1}},\;\star,{0_{s_2}}, \star ,{0_{s_3}}} \right),$$ without loss of generality, suppose that the first coordinate of $c\left( x \right)$ is 1. We denote that
	
	$$c\left( x \right) = 1 + {a_1}x + {a_2}{x^2} + {a_3}{x^l} + {a_4}{x^t}.$$
	
Suppose that $l \equiv i({\bmod~3}) $ and $t \equiv j({\bmod~3}),\, i,j=0,1,2 $, since $l$ and $t$ positionally equivalent, we need to discuss six situations here. Similar to ${\bf{Subcase\; 1.1}}$, we summarize all $i$ and $j$ of ${\bf{Subcase\; 1.2}}$ in the following Table \ref{table-3}.
\end{description}

\begin{table}[h]\label{table3}
	\begin{center}
				\begin{minipage}{\textwidth}
			\caption{Summary of ${\bf{Subcase\; 1.2}}$}\label{table-3}
			\begin{center}
		
			\begin{tabular}{@{}lllll@{}}
					\toprule				
				$i$ & $j$ & Conditions  & \makecell{Results}& Contradictory \\
				\midrule
				0 & 0 &$\left[\kern-0.15em\left[ 1 
				\right]\kern-0.15em\right],\left[\kern-0.15em\left[ 2 
				\right]\kern-0.15em\right]$ & \makecell{$ a_1-a_2=0 $,\\$a_1+a_2=0$ }& $a_1,a_2 \in {\Bbb F}_p^{\text{*}}$\\  
				\midrule
				0 & 1 & $\left[\kern-0.15em\left[ 1 \right]\kern-0.15em\right],\left[\kern-0.15em\left[ 2 \right]\kern-0.15em\right]$ & \makecell{$ a_1-a_2+a_4=0 $,\\$a_1+a_2+a_4=0$ } & $a_2 \in {\Bbb F}_p^{\text{*}}$\\
				\midrule  
				0 & 2 & $\left[\kern-0.15em\left[ 1 \right]\kern-0.15em\right],\left[\kern-0.15em\left[ 2 \right]\kern-0.15em\right]$ & \makecell{$ a_1+a_2+a_4=0 $,\\$a_1-a_2-a_4=0$ } & $a_1 \in {\Bbb F}_p^{\text{*}}$\\
				\midrule
				1 & 1 & $\left[\kern-0.15em\left[ 1 \right]\kern-0.15em\right],\left[\kern-0.15em\left[ 2 \right]\kern-0.15em\right]$ & $\makecell{3=0}$ & $p\ne 3$ \\
				\midrule 
				1 & 2 & $\left[\kern-0.15em\left[1\right]\kern-0.15em\right],\left[\kern-0.15em\left[ 2 \right]\kern-0.15em\right]$ & $\makecell{3=0}$ & $p\ne 3$ \\
				\midrule
				2 & 2 & $\left[\kern-0.15em\left[1\right]\kern-0.15em\right],\left[\kern-0.15em\left[ 2 \right]\kern-0.15em\right]$ &\makecell{$3=0$} & $p\ne 3 $\\
			
			\end{tabular}

			\footnotetext{ Conditions $\left[\kern-0.15em\left[ 1 
					\right]\kern-0.15em\right]$ and $\left[\kern-0.15em\left[ 2 \right]\kern-0.15em\right]$  represent $c\left( 1 \right) = c\left( { \omega} \right) =c\left( { \omega^2} \right) = 0$ and $c\left( 1 \right) + c\left( { \omega} \right) +c\left( { \omega^2} \right) = 0$,  respectively.}
			\end{center}
				\end{minipage}
			
	\end{center}

\end{table}

\begin{description}

	\item[\textbf{Case II.}] If there are codewords with Hamming weight $\omega_H=6$ and symbol-pair weight $\omega_p=8$, then its certain
	cyclic shift must be one of the following forms
	$$ \left(  \star , \star ,\star ,\star ,\star ,{0_{s_1}}, \star,{0_{s_2}} \right),\;\;$$ 
	$$ \left(  \star ,\star ,\star , \star ,{0_{s_1}}, \star ,\star ,{0_{s_2}}\right)\;\;\;$$ or 
	$$ \left(  \star ,\star ,\star , {0_{s_1}},\star ,\star , \star ,{0_{s_2}}\right),$$
	where each $ \star $ denotes an element in ${\Bbb F}_p^{\text{*}}$ and  ${0_{s_1}}$, ${0_{s_2}}$, ${0_{s_3}}$ are all-zero vectors.
	
	${\bf{Subcase\; 2.1}}$ For the subcase of $$ \left( { \star , \star , \star,\star , \star ,{0_{s_1}},\star,{0_{s_2}}} \right),$$ without loss of generality, suppose that the first coordinate of $c\left( x \right)$ is 1. We denote 
	
	$$c\left( x \right) = 1 + {a_1}x + {a_2}{x^2} + {a_3}{x^{ 3} }+{a_4}{x^4} + {a_5}{x^{ t} } .$$
	\begin{enumerate}

		\item When $t \equiv 0\left( {\bmod~3} \right)$, it be derived from $c\left( 1 \right) = c\left( { \omega} \right) =c\left( { \omega^2} \right) = 0$ and $ c\left( 1 \right) + c\left( { \omega} \right) +c\left( { \omega^2} \right) = 0$ that
		\begin{equation*}  
			\left\{  
			\begin{array}{lr}  
				{a_1}+{a_2}+{a_4}=0, &  \\  
				{a_1}-{a_2}+{a_4}=0, &
			\end{array}  
			\right.  
		\end{equation*}
		then we have $2a_2=0 $, which is impossible.				
		
		\item When $t \equiv 1\left( {\bmod~3} \right)$, similar to $t \equiv 0\left( {\bmod~3} \right)$, $2a_2=0 $ can be obtained
		from	$c\left( 1 \right) = c\left( { \omega} \right) =c\left( { \omega^2} \right) = 0$ and $ c\left( 1 \right) + c\left( { \omega} \right) +c\left( { \omega^2} \right) = 0$.	
		
		\item  When $t \equiv 2\left( {\bmod~3} \right)$, it follows from $c\left( 1 \right) = c\left( { \omega} \right) =c\left( { \omega^2} \right) = 0$ that
		\begin{equation*}  
			\left\{  
			\begin{array}{lr}  
				1+{a_1}+{a_2}+{a_3}+{a_4}+{a_5}=0, &  \\  
				1+{a_1}\omega+{a_2}\omega^2+{a_3}+{a_4}\omega+{a_5}\omega^2=0, &\\  
				1+{a_1}\omega^2+{a_2}\omega+{a_3}+{a_4}\omega^2+ {a_5}\omega=0. &    
			\end{array}  
			\right.  
		\end{equation*}
		By solving the system, we have $1+a_3=0,a_1+a_4=0,a_2+a_5=0$. Then $c'\left( 1 \right)= c'\left( \omega \right) =0$  indicates
		\begin{equation*}  
			\left\{  
			\begin{array}{lr}  
				{a_1} + 2{a_2} +3 {a_3} + 4{a_4}+ t{a_5}  = 0, &\\
				{a_1} + 2{a_2}\omega +3 {a_3}\omega^2 + 4{a_4}+ t{a_5} \omega^2 = 0, &    
			\end{array}  
			\right.  
		\end{equation*}
		which means that ${a_4}={a_3}\omega $ and ${a_5} = {{3{a_3}\omega^2 } \over {t - 2}}$ (since $t \equiv 2\left( {\bmod~3} \right)$, then $p$ is not a divisor of $t-2$, otherwise $t-2 \ge 3p$).
		
		By $ c''\left( 1 \right) =0$, we have $t = 3+2\omega$. Together with  ${a_3}\omega = {a_4}$, ${a_5} = {{3{a_3}\omega^2 } \over {t - 2}}$ and $c'''\left( 1 \right) = 0$, then one can derive that
		$$ 6+24\omega+3t(t-1)\omega^2=0,$$ we have $$2+8\omega+(3+2\omega)(2+2\omega)\omega^2=0,$$ combining with $\omega^2=-1-\omega$, we can obtain $3\omega^2=0$, which is impossible.
		
	\end{enumerate}

	${\bf{Subcase\; 2.2}}$ For the subcase of $$ \left( { \star , \star , \star,\star ,{0_{s_1}}, \star ,\star,{0_{s_2}}} \right),$$ without loss of generality, suppose that the first coordinate of $c\left( x \right)$ is 1. We denote that
	
	$$c\left( x \right) = 1 + {a_1}x + {a_2}{x^2} + {a_3}{x^{ 3} }+{a_4}{x^t} + {a_5}{x^{ t+1} } .$$
	\begin{enumerate}

		\item  When $t \equiv 0\left( {\bmod~3} \right)$, similar to $t \equiv 0\left( {\bmod~3} \right)$ in ${\bf{Subcase\; 2.1}}$,	$2a_2=0 $ can be derived from
		$c\left( 1 \right) = c\left( { \omega} \right) =c\left( { \omega^2} \right) = 0$.	
		
		\item  When $t \equiv 1\left( {\bmod~3} \right)$, 
		similar to $t \equiv 2\left( {\bmod~3} \right)$ in ${\bf{Subcase\; 2.1}}$, $1+a_3=0,a_1+a_4=0,$ and $a_2+a_5=0$ can be obtained from $c\left( 1 \right) = c\left( { \omega} \right) =c\left( { \omega^2} \right) = 0$.
		
		Then $ {a_5}={a_4}\omega$ and ${a_5} = {{3{a_3}\omega^2 } \over {t - 1}}$ can be derived from $c'\left( 1 \right)= c'\left( \omega \right) =0$. 
		$c''\left( 1 \right)=0$ means $t=-\omega$.
		
		Finally, combined with $c'''\left( 1 \right)=0$, we have $\omega^2-\omega = 0$, a contradiction.			
		
		\item When $t \equiv 2\left( {\bmod~3} \right)$, similar to $t \equiv 0\left( {\bmod~3} \right)$ in ${\bf{Subcase\; 2.1}}$,	$2a_2=0 $ can be derived from
		$c\left( 1 \right) = c\left( { \omega} \right) =c\left( { \omega^2} \right) = 0$.	
	\end{enumerate}
	${\bf{Subcase\; 2.3}}$ For the subcase of $$ \left( { \star , \star , \star,{0_{s_1}},\star , \star ,\star,{0_{s_2}}} \right),$$ without loss of generality, suppose that the first coordinate of $c\left( x \right)$ is 1. We denote that
	
	$$c\left( x \right) = 1 + {a_1}x + {a_2}{x^2} + {a_3}{x^{ t} }+{a_4}{x^{t+1}} + {a_5}{x^{ t+2} } .$$
	\begin{enumerate}

	\item 	When $t \equiv 0\left( {\bmod~3} \right)$, it follows from $c\left( 1 \right) = c\left( { \omega} \right) =c\left( { \omega^2} \right) = 0$ that
	\begin{equation*}  
		\left\{  
		\begin{array}{lr}  
			1+{a_1}+{a_2}+{a_3}+{a_4}+{a_5}=0, &  \\  
			1+{a_1}\omega+{a_2}\omega^2+{a_3}+{a_4}\omega+{a_5}\omega^2=0, &\\  
			1+{a_1}\omega^2+{a_2}\omega+{a_3}+{a_4}\omega^2+ {a_5}\omega=0. &    
		\end{array}  
		\right.  
	\end{equation*}
	By solving the system, we have  $1+a_3=0,a_1+a_4=0,a_2+a_5=0$.	
	Then $c'\left( 1 \right)= c'\left( \omega \right) =0$  indicates
	\begin{equation*}  
		\left\{  
		\begin{array}{lr}  
			{a_1} + 2{a_2} +t {a_3} + \left( t+1\right) {a_4}+ \left( t+2\right) {a_5}  = 0, &\\
			{a_1} + 2{a_2}\omega +t {a_3}\omega^2 + \left( t+1\right) {a_4}+ \left( t+2\right){a_5} \omega = 0, &    
		\end{array}  
		\right.  
	\end{equation*}
	which means that $  {a_5} ={a_4}\omega=a_3\omega^2$. Then $c''\left( 1 \right)= 0$ implies that
	$$2{a_2} +t\left( t-1\right)  {a_3} +t \left( t+1\right) {a_4}+ \left( t+1\right) \left( t+2\right) {a_5}  = 0,$$
	which implies $t(3\omega^2+\omega-1=0) $. Since  $t \equiv 0\left( {\bmod 3} \right), \omega=-\omega^2-1$ and the code length $3p$, we have $2(\omega^2-1)=0$, a contradiction.				
	
	\item 	When $t \equiv 1\left( {\bmod~3} \right)$, with arguments similar to the previous $t \equiv 0\left( {\bmod~3} \right)$, by $c\left( 1 \right) = c\left( { \omega} \right) =c\left( { \omega^2} \right) = 0$ and $c'\left( 1 \right)= c'\left( \omega \right) =0$, we have ${a_4} = {a_3}\omega= {{{a_5}\omega^2{\left( t +2\right) } } \over {t-1}}$. Together with $ c''\left( 1 \right)  = 0$, $\omega^2-\omega =0$ can be derived, which is impossible.

	\item 	When $t \equiv 2\left( {\bmod~3} \right)$, with arguments similar to the previous $t \equiv 0\left( {\bmod~3} \right)$, we can derive ${a_5}={a_4}\omega =  {{{a_3}\omega^2{\left( t - 2\right) } } \over {t+1}}$ from  $c'\left( 1 \right)= c'\left( \omega \right) =0$ and $c\left( 1 \right) = c\left( { \omega} \right) =c\left( { \omega^2} \right) = 0$. Together with $ c''\left( 1 \right)  = 0$, we have $t =p+1$. It follows from $c'\left( 1 \right)= c'\left( \omega \right) =0$ that
	 	\begin{equation*}  
	 	\left\{  
	 	\begin{array}{lr}  
	 		(t-1) {a_3} + \left( t+1\right) {a_4}+ \left( t+1\right) {a_5}  = 0, &\\
	 		(t-1) {a_3}\omega + \left( t+1\right) {a_4}\omega^2+ \left( t+1\right){a_5}  = 0. &    
	 	\end{array}  
	 	\right.  
	 \end{equation*}
	 Combining $t =p+1$,  we have $\omega^2=1$, which contradicts that $\omega$ is primitive 3-th root of unity in ${\Bbb F}_{p}$.
		\end{enumerate}
	
\end{description}

\begin{description}
	\item[\textbf{Case III.}] If there are codewords in $\mathcal{C}$ with Hamming weight 7 and symbol-pair weight 8, then its certain cyclic shift must have the form $$ \left( { \star , \star , \star , \star ,\star , \star , \star ,{0_s}} \right),$$ where each $ \star $ denotes an element in  ${\Bbb F}_p^{\text{*}}$ and ${0_s}$ is all-zero vector. Without loss of generality, suppose that the first coordinate of $c\left( x \right)$ is 1. We denote
	
	$$c\left( x \right) = 1 + {a_1}x + {a_2}{x^2} + {a_3}{x^3} + {a_4}{x^4} + {a_5}{x^5}+ {a_6}{x^6}.$$
	
	This leads to $\deg\left( {c\left( x \right)} \right) = 6 < 7 = \deg\left( {g\left( x \right)} \right)$.

\end{description}	

As a result,  $\mathcal{C}_{421}$ is an MDS ${\left( {3p,{\text{\;}}9} \right)_p}$ symbol-pair code.	
\end{proof}
In what follows, let's determine the minimum symbol-pair distance for $\mathcal{C}_{r_{1}r_{2}r_{3}}$ in the previous paper by using the following propositions.

\begin{Proposition}\label{proposition of C_{rrr}}
	\item 
\begin{enumerate}
\item The minmum distance of $\mathcal{C}_{0r_{2}0}$ is $d_{p}=4$, when $r_{2} \ge 2$.
\item The minmum distance of $\mathcal{C}_{210}$ is $d_{p}=5$.
\item The minmum distance of $\mathcal{C}_{r_{1}r_{2}0}$ is $d_{p}=6$, when   $r_{1}+r_{2} \ge 4$ and $r_{1},r_{2} \in {\Bbb F}_p^{\text{*}}$.
\item The minmum distance of $\mathcal{C}_2{r_{2}r_{3}}$ is $d_{p}=6$, when  $2\le r_{2}+r_{3}\le 4$.
\end{enumerate}	
\end{Proposition}
\begin{proof}

 For $\mathcal{C}_{0r_{2}0}=\left\langle {{(x - \omega )}^{r_{2}}} \right\rangle$, Lemma \ref{lemma 3.1} shows $d_{H}(\mathcal{C}_{0r_{2}0})=2$, then one can obtain $d_{p}(\mathcal{C}_{0r_{2}0})=4$ by Lemma \ref{lemma 3.3 }.

Theorem \ref{theorem 3.6} proves the minmum distance of $\mathcal{C}_{210}$ is $d_{p}=5$.
	
Since $\mathcal{C}_{r_{1}r_{2}0}$ is a subcode of Theorem \ref{theorem 3.6} and Lemma \ref{lemma 3.1} means $d_{H}(\mathcal{C}_{r_{1}r_{2}0})=3$. Then we have $d_{p}(\mathcal{C}_{r_{1}r_{2}0})=6$.
\end{proof}
Through Proposition \ref{proposition of C_{rrr}}, we can obtain all MDS and AMDS symbol-pair codes of length $3p$ with minimum symbol-pair distance of 4 to 6. Next, we look for MDS and AMDS symbol-pair codes with the minimum symbol pair distance $d_{p}=7$.
\begin{Proposition}\label{proposition of C_{r11}}
The minmum distance of $\mathcal{C}_{r_{1}11}$ is $d_{p}=7$, when   $3\le r_{1}\le p-1$.
\end{Proposition}
\begin{proof}
Reference \cite{Chen-2017-CIT} proved $\mathcal{C}_{311}$ is an MDS symbol-pair code with the minmum symbol-pair distance $d_{p}=7$.  	
 
Since $\mathcal{C}_{r_{1}11}$ is a subcode of $\mathcal{C}_{311}$, we have  $d_{p}(\mathcal{C}_{r_{1}11})\ge 7$. 
Then the minmum symbol-pair distance $d_{p}(\mathcal{C}_{r_{1}11})= 7$ can be obtained by $\omega_p(c(x))=7$, where $c(x)=1-x-x^p+x^{2p+1}$ is a codeword of  $\mathcal{C}_{r_{1}11}$.
\end{proof}
From this proposition, we can deduce that $\mathcal{C}_{r_{1}11}$ is an AMDS symbol-pair code with the minimum symbol-pair distance 7 if and only if $ r_{1}=4$. Through the following proposition, we find the AMDS symbol-pair codes with the minimum symbol-pair distance 8.
\begin{Proposition}\label{proposition of C_{3rr}}
	The minmum distance of $\mathcal{C}_{3r_{2}r_{3}}$ is $d_{p}=8$, when   $3\le r_{2}+r_{3}\le6$.
\end{Proposition}
\begin{proof}
	Similar to the previous Proposition \ref{proposition of C_{r11}}, Reference \cite{Chen-2017-CIT} proved that $\mathcal{C}_{321}$ is an MDS $(3p,8)_{p}$ symbol-pair code. Since $3\le r_{2}+r_{3}\le6$, we have $\mathcal{C}_{3r_{2}r_{3}}$ is subcode of $\mathcal{C}_{321}$ and $d_{p}(\mathcal{C}_{3r_{2}r_{3}})\ge 8$.
	
	By Lemma \ref{lemma 3.1}, we can deduce that the minimum Hamming distance of $\mathcal{C}_{3r_{2}r_{3}}$ is 4, which implies  $d_{p}(\mathcal{C}_{3r_{2}r_{3}})\le 8$. 
	
	Therefore, the minmum distance of $\mathcal{C}_{3r_{2}r_{3}}$ is $d_{p}=8$, when   $3\le r_{2}+r_{3}\le6$.
\end{proof}
Proposition \ref{proposition of C_{3rr}} indicates that $\mathcal{C}_{3r_{2}r_{3}}$ is an AMDS symbol-pair code with $d_{p}=8$, if  $r_{2},r_{3}\ne 0$ and $r_{2}+r_{3}=4$. Proposition \ref{Proposition C_{421}} shows that $\mathcal{C}_{421}$ is an MDS symbol-pair code with mimmum symbol-pair distance 9. In what follows, we will give all AMDS symbol-pair codes with minimum symbol-pair distance 9.

\begin{Proposition}\label{proposition of C_{rr1}}
	The minmum distance of $\mathcal{C}_{r_{1}r_{2}1}$ is $d_{p}=9$, when   $4\le r_{1}\le p-1$ and $2\le r_{2}\le p-1$.
\end{Proposition}
\begin{proof}
With arguments similar as the proof of Proposition	\ref{proposition of C_{r11}}, since $\mathcal{C}_{421}$ is an MDS symbol-pair code with the minmum symbol-pair distance 9, we can deduce that the minimum symbol-pair distance of $\mathcal{C}_{r_{1}r_{2}1}$ is $d_{p}(\mathcal{C}_{r_{1}r_{2}1})\ge 9$. Next, we will prove that there are symbol-pair codewords with symbol-pair weight $d_{p}=9$ in $\mathcal{C}_{r_{1}r_{2}1}$.
	
	For the codeword 
	$$c(x)=(x-1)^p(x-\omega)^p(x-\omega^2)= {x^{2p + 1}} - {\omega ^2}{x^{2p}} + {\omega ^2}{x^{p + 1}} - \omega {x^p} + \omega x - 1,$$
it is easy to verify that $c(x)$ is a codeword polynomial of $\mathcal{C}_{r_{1}r_{2}1}$ with the symbol-pair weight $\omega_p(c(x))=9$.
Thus, we have $d_{p}(\mathcal{C}_{r_{1}r_{2}1})=9$.
\end{proof}
From Proposition \ref{proposition of C_{rr1}}, We can determine that $\mathcal{C}_{r_{1}r_{2}1}$ is an AMDS symbol-pair code with $d_{p}=9$, if and only if one of $r_{1}=5,r_{2}=2$ and $r_{1}=4,r_{2}=3$ is satisfied. The following proposition will show that symbol-pair codes with minimum symbo-pair distance $d_{p}=10$.
\begin{Proposition}\label{proposition of C_{rr2}}
	The minmum distance of $\mathcal{C}_{r_{1}r_{2}r_{3}}$ is $d_{p}=10$, when   $ r_{1}$, $ r_{2}$ and $ r_{3}$ meet any of the following two conditions
	\begin{enumerate}
		\item $4\le r_{1}\le p-1$ and  $ r_{2}= r_{3}=2$,
		\item $r_{1}= 4$ and  $4\le r_{2}+ r_{3}\le 8$.
	\end{enumerate}
\end{Proposition}
\begin{proof}
	For $4\le r_{1}\le p-1$ and  $ r_{2}= r_{3}=2$, with arguments similar as the proof of Proposition	\ref{proposition of C_{r11}}, reference \cite{Ma-2022-DCC} proved $\mathcal{C}_{422}$ is an MDS symbol-pair code with $d_{p}=10$.
	Since $\mathcal{C}_{r_{1}r_{2}r_{3}}$ is a subcode of $\mathcal{C}_{422}$, we have $d_{p}(\mathcal{C}_{r_{1}22})\ge 10$. 
	
	Note that the codeword polynomial
	$$c(x)=1-x^2+2x^{p+1}+x^{p+2}-x^{2p}-2x^{2p+1}$$
	is a codeword of $\mathcal{C}_{r_{1}22}$ and $\omega_p(c(x))=10$. Therefore, we derive the minmum symbol-pair distance of $\mathcal{C}_{r_{1}22}$ is 10.
	
	For $r_{1}= 4$ and  $4\le r_{2}+ r_{3}\le 8$, since $\mathcal{C}_{4r_{2}r_{3}}$ is a subcode of $\mathcal{C}_{422}$, we have $d_{p}(\mathcal{C}_{4r_{2}r_{3}})\ge 10$.  
	
	However, Lemma \ref{lemma 3.1} shows that $d_{H}(\mathcal{C}_{4r_{2}r_{3}})=5$, which implies that $d_{p}(\mathcal{C}_{4r_{2}r_{3}})\le 10$. Thus, we can deduce $d_{p}(\mathcal{C}_{4r_{2}r_{3}})= 10$.
\end{proof}
From Proposition \ref{proposion 3.9} to Proposition \ref{proposition of C_{rr2}}, our proof contains all minimum symbol-pair distances of repeated-root cyclic codes $\mathcal{C}_{r_{1}r_{2}r_{3}}$, where the degree of the generator polynomials of $\mathcal{C}_{r_{1}r_{2}r_{3}}$ is less than 11. Under the above constraints, we can easily deduce that there is no repeated-root cyclic code with a minimum symbol-pair distance  $d_{p}=11$. We have the following corollary.

\begin{corollary}
The repeated-root cyclic code $\mathcal{C}_{r_{1}r_{2}r_{3}}$ must not be the MDS and the AMDS symbol-pair code with a minimum symbol-pair distance $d_{p}=11$. 	
\end{corollary}

Proposion \ref{proposition of C_{rr2}} shows that the condition does not satisfy Theorem \ref{theorem of 3p MDS}, when $r\ge 5$ and $r_{1} = r_{2} + r_{3} + 1$. Next, we use an example to illustrate that the conditions of  Theorem \ref{theorem of 3p MDS}  are also no longer applicable, when $r>5$ and $r_{1} = r_{2} + r_{3} $.
\begin{example}
Let $\mathcal{C}$ and  be a repeated-root cyclic code in ${{\Bbb F}_7}\left[ {\text{\textit{x}}} \right]/\left\langle {{x^{21}} - 1} \right\rangle$ and the generator ploynomial of $ \mathcal{C} $  is
$$g\left( x \right) = (x-1)^6(x - 2)^3(x - 4 )^3,$$
where $\omega=2$ is a 3-th primitive element in ${{\Bbb F}_7}$ and $2^2=4$. 

Then we have the minmum Hamming distance $d_H=7$ by a magma progarm.
Reference \cite{Ma-2022-DCC} shows the symbol-pair distance $d_p\ge 12$. The magma program also shows that both vectors
$$\textbf{a}=[0\; 0\; 1\; 0\; 0\; 0\; 0\; 0\; 0\; 6\; 6\; 3\; 3\; 3\; 4\; 4\; 4\; 5\; 1\; 1\; 1]$$
and
$$\textbf{b}=[0\; 0\; 0\; 0\; 0\; 0\; 1\; 0\; 0\; 1\; 1\; 1\; 5\; 4\; 4\; 4\; 3\; 3\; 3\; 6\; 6]$$
are in $\mathcal{C}$. We add the two vectors $\textbf{a+b}$ to get the new vector $\textbf{c}$, 
$$\textbf{c}=[0\; 0\; 1\; 0\; 0\; 0\; 1\; 0\; 0\; 0\; 0\; 4\; 1\; 0\; 1\; 1\; 0\; 1\; 4\; 0\; 0],$$
which is also in $\mathcal{C}$. We can easily deduce $\omega_p(\textbf{c})=13$.
 Therefore, $\mathcal{C}$ is not an MDS symbol-pair code.	
\end{example}

\subsection[Constructions of AMDS Symbol-Pair Codes]{AMDS Symbol-Pair Codes}
In this subsection, for length $lp$ and $4p$, we propose two new classes of AMDS symbol-pair codes. For preparation, we define the following notation.

Let $\mathcal{C}_{1}$ and $\mathcal{C}_{2}$ be the cyclic codes in ${{\Bbb F}_p}\left[ {\text{\textit{x}}} \right]/\left\langle {{x^n} - 1} \right\rangle$. The generator ploynomial of $\mathcal{C}_{1}$ is $$g_{1}\left( x \right) = {\left( {x - 1} \right)^2}\left( {{x} + 1} \right)\left( {{x} - \omega} \right),$$ where $\omega$ is a primitive $l$-th  root of unity in ${{\Bbb F}_p}$. The generator ploynomial of $\mathcal{C}_{2}$ is $$g_{2}\left( x \right) = {\left( {x - 1} \right)^3}\left( {x + 1} \right)^2\left( {{x}^{2} +1} \right),$$ where $\omega$ is a primitive $ 4$-th  root of unity in ${{\Bbb F}_p},$ when $p \equiv 1\left( {\bmod~4} \right)$, and $\omega$ is a primitive $ 4$-th  root of unity in ${{\Bbb F}_{{p^2}}}\backslash {{\Bbb F}_p}$, when $p \equiv 3\left( {\bmod~4} \right)$.

Now we present a class of AMDS symbol-pair codes with $n = lp$ and minimum symbol-pair distance 7. 
\begin{theorem}\label{theorem3.9}
$\mathcal{C}_{1}$ is an AMDS ${\left( {lp,\;7} \right)_p}$ symbol-pair code, if $ l $ odd and $ l \geqslant 3 $.
\end{theorem}
\begin{proof}
$\mathcal{C}_{1}$ is the cyclic code in ${{\Bbb F}_p}\left[ {\text{\textit{x}}} \right]/\left\langle {{x^n} - 1} \right\rangle$ generated by $$g_{1}\left( x \right) = {\left( {x - 1} \right)^4}\left( {{x} -\omega} \right)\left( {{x} - {\omega}^2} \right).$$

By Lemma \ref{lemma 3.1}, one can derive that $\mathcal{C}_{1}$ is an $\left[ {lp,\;lp - 6,\;4} \right]$ repeated-root cyclic codes code over ${{\Bbb F}_p}$. Lemma \ref{lemma 3.3 } yields that ${d_p} \geqslant 6$, since $\mathcal{C}_{1}$ is not an MDS cyclic code. To prove that $\mathcal{C}_{1}$ is an AMDS symbol-pair code with the minmum symbol-pair distance 7, it is sufficient to verify that there is no a codeword in $\mathcal{ C}_{1}$ with the symbol-pair weight 6. 

If there are codewords in $\mathcal{C}_{1}$ with  Hamming weight 5 and symbol-pair weight 6, then its certain cyclic shift must have the form
$$ \left( { \star ,\; \star ,\; \star ,\; \star ,\; \star ,{0_s}} \right),$$
where each $ \star $ denotes an element in  ${\Bbb F}_p^{\text{*}}$ and ${0_s}$ is all-zero vector. Without loss of generality, suppose that the first coordinate of $c\left( x \right)$ is 1. We denote that
$$c\left( x \right) = 1 + {a_1}x + {a_2}{x^2} + {a_3}{x^3} + {a_4}{x^4} ,$$
This leads to $\deg\left( {c\left( x \right)} \right) =4 < 6 =\deg\left( {g\left( x \right)} \right)$.

\noindent If $c \in \mathcal{ C}$ has the symbol-pair weight 6 with the Hamming weight 4,

then its certain cyclic shift must have the forms
$$\left( { \star ,\star , \star ,{0_{s_1}}, \star ,{0_{s_2}}} \right)$$ or
$$ \left( { \star , \star ,{0_{s_1}},\star , \star ,{0_{s_2}}} \right),$$
where each $ \star $ denotes an element in ${\Bbb F}_p^{\text{*}}$ and ${0_{s_{1}}}$, ${0_{s_{2}}}$ are all-zero vectors.
\begin{description}	
	\item[\textbf{Case I.}] For the case of 	$$ \left( { \star ,\star , \star ,{0_{s_1}}, \star ,{0_{s_2}}} \right),$$
	without loss of generality, we denote a codeword polynomial
	$$c\left( x \right) = 1 + {a_1}x + {a_2}{x^2} + {a_3}{x^t}.$$
	It follows from $c'\left( 1 \right) = c''\left( 1 \right) = 0$ that
	\begin{equation*}  
		\left\{  
		\begin{array}{lr}    
			{a_1} + 2{a_2} + t{a_3} = 0,   &\\  
			2{a_2} + t(t-1){a_3} = 0.   & 
		\end{array}  
		\right.  
	\end{equation*}
	By solving the system, we have $ t(t-2){a_3} ={a_1} $. Then one induces that $t - 2 \ne kp,k<l $ and $t \ne kp,k<l $ . By $c''\left( 1 \right) = 0$, we can conclude that $ t \ne kp,k<l $ and $t - 1 \ne kp,k<l $. However,  $ c'''\left( 1 \right)= 0 $ implies $$ t(t-1)(t-2){a_3} = 0, $$ then one can be checked that $$ t(t-1)(t-2)=kp, $$ which contradicts $ t \ne kp $, $t - 1 \ne kp$ and $t - 2 \ne kp .$
	
	\item[\textbf{Case II.}] For the case of 	$$ \left( { \star , \star ,{0_{s_1}},\star , \star ,{0_{s_2}}} \right),$$
	without loss of generality, we denote
	$$c\left( x \right) = 1 + {a_1}x + {a_2}{x^2} + {a_3}{x^t} + {a_4}{x^{t + 1}}.$$
	It follows from $c\left( 1 \right) =c'\left( 1 \right)= 0$ that
	\begin{equation*}  
		\left\{  
		\begin{array}{lr}  
			1 + {a_1} + {a_2} + {a_3} = 0,   &  \\  
			{a_1} + t{a_2} + (t+1){a_3} = 0,  & 
		\end{array}  
		\right.  
	\end{equation*}
	one can derive that  $$ (t-1){a_2} + t{a_3}-1 = 0. \; $$
	By $c''\left( 1 \right) = 0$, we have
	$$ t(t-1){a_2} + t(t+1){a_3} = 0. $$
	This leads to $ t({a_3}+1) = 0 $. 
	Therefore, we have $ t=kp,0<k<l $ or $ a_{3}= -1 $.
	
	If $ t=kp,0<k<l $, then 
	\begin{equation*}  
		\left\{  
		\begin{array}{lr}  
			1 + {a_1} + {a_2} + {a_3} = 0, &  \\  
			{a_1} + {a_3} = 0.  & 
		\end{array}  
		\right.  
	\end{equation*}
	This indicates $ {a_1} =- {a_3} $	and $ {a_2}=-1 $. Combined with  $c\left( \omega  \right) = c\left( {\omega ^2} \right) = 0$, we have	
	\begin{equation*}  
		\left\{  
		\begin{array}{lr}  
			{a_1}\omega \left( {{\omega ^t} - 1} \right) = 1 - {\omega ^t}, &  \\  
			{a_1}\omega^2 \left( {{\omega ^{2t}} - 1} \right) = 1 - {\omega ^{2t}}.  & 
		\end{array}  
		\right.  
	\end{equation*}
	By solving the system,we have ${\omega ^{2t}} = 1 $, which contradicts $ l $ odd.
	
	If $ a_{3}= -1 $, by $ c(1)=0 $, we can obtain that $ {a_1} =- {a_2} $. Combined with  $c\left( \omega  \right) = c\left( {\omega ^2} \right) = 0$, we have	
	\begin{equation*}  
		\left\{  
		\begin{array}{lr}  
			{a_1}\omega \left( {{\omega ^{t-1}} - 1} \right) = 1 - {\omega ^{t+1}}, &  \\  
			{a_1}\omega^2 \left( {{\omega ^{2t-2}} - 1} \right) = 1 - {\omega ^{2t+2}}.  & 
		\end{array}  
		\right.  
	\end{equation*}
	
	\noindent Since $\omega$ is a primitive $ l$-th root of unity, then \[\omega \left( {{\omega ^{t - 1}} + 1} \right)\left( {1 - {\omega ^{t + 1}}} \right) = 1 - {\omega ^{2t + 2}}.\]
	This implies that $ \omega ^t = 1 $.
	Thus, $ a_{1}= -1, a_{2}= 1 $.
	By $$ t(t-1){a_2} + t(t+1){a_3} = 0,$$ we have $ 2t=kp $, which contradicts $ \omega ^t = 1 $.

\end{description}

In order to prove that $\mathcal{ C}_{1}$ is an AMDS symbol-pair code, we need to find a codeword with the symbol-pair weight 7. Since 
$$ c(x)=(x^{p}-1)(x^{p-1}-1)=x^{2p-1}-x^{p}-x^{p-1}+1 $$
 is a codeword of $\mathcal{ C}_{1}$ and $ {\omega _p}\left( {c(x)} \right)=7 $, 	$\mathcal{C}_{1}$ is an AMDS ${\left( {lp,\;7} \right)_p}$ symbol-pair code. 
\end{proof}

In what follows, we construct a class of AMDS symbol-pair codes with $n = 4p$ and minimum symbol-pair distance 8.
\begin{theorem}\label{theorem 3.10}
$\mathcal{C}_{2}$ is an AMDS ${\left( {4p,\;8} \right)_p}$ symbol-pair code.
\end{theorem}
\begin{proof}
Since $\mathcal{C}_{2}$ is the cyclic code in ${{\Bbb F}_p}\left[ {\text{\textit{x}}} \right]/\left\langle {{x^n} - 1} \right\rangle$ generated by $$g_{2}\left( x \right) = {\left( {x - 1} \right)^3}\left( {x + 1} \right)^2\left( {{x}^{2} +1} \right).$$ For any codeword $c \in \mathcal{C}$, we have $$c\left( 1 \right) = c\left( { - 1} \right) = c\left( \omega  \right) = c\left( { - \omega } \right) = c'\left( 1 \right) = c''\left( 1 \right) = 0,$$ where $\omega $ is a primitive 4-th root of unity in ${{\Bbb F}_{{p^2}}}\backslash {{\Bbb F}_p}$ or ${\Bbb F}_{p}$.

\noindent By Lemma \ref{lemma 3.1}, $\mathcal{C}_{2}$ is an $\left[ {4p,\;4p - 7,\;4} \right]$ cyclic code over ${{\Bbb F}_p}$. Since $\mathcal{C}_{2}$ is not an MDS cyclic code, Lemma \ref{lemma 3.3 } yields that ${d_p} \geqslant 6$.

When  $p \equiv 3\left( {\bmod~4} \right)$, since $\mathcal{C}_{2}$ is a subcode of Lemma \ref{lemma 3.5 } and  it is shown in the proof of Lemma 2.5 that there does not exist a codeword in $\mathcal{C}_{2}$ with symbol-pair weight 6 and 5. 

When $p \equiv 1\left( {\bmod~4} \right)$, we are left to show that there
are no codewords of $\mathcal{C}_{2}$ with symbol-pair weight 6 and Hamming weight  4 or 5.
\begin{description}	
	\item[\textbf{Case I.}]	
Assume that $c\left( x \right)$ has the Hamming weight 5 and the symbol-pair weight 6. Then its certain cyclic shift must have the form
	$$\left( { \star ,\; \star ,\; \star ,\; \star,\; \star,\;{0_{s}}} \right),$$
	where each $ \star $ denotes an element in ${\Bbb F}_p^{\text{*}}$ and ${0_{s}}$ is all-zero vector. Without loss of generality, suppose that the first coordinate of $c\left( x \right)$ is 1. Then $$c\left( x \right) = 1 + {a_1}x + {a_2}{x^2} + {a_3}{x^3}+ {a_4}{x^4}$$ is a codeword of $\mathcal{C}_{2}$. However, this leads to $\deg\left( {c\left( x \right)} \right) =4 < 7 = \deg\left( {g\left( x \right)} \right)$.
	
	\item[\textbf{Case II.}] Assume that $c\left( x \right)$ has the Hamming weight 4 and the symbol-pair weight 6. Then its certain cyclic shift must have one of the forms
	$$ \left( { \star ,\; \star ,\; \star ,{0_{s_1}}, \star ,{0_{s_2}}} \right)$$ or  
	$$\left( { \star ,\; \star ,\;{0_{s_1}}, \star , \,\star ,{0_{s_2}}} \right),$$
	where each $ \star $ denotes an element in ${\Bbb F}_p^{\text{*}}$ and ${0_{s1}},{0_{s2}}$ are all-zero vectors. Without loss of generality, suppose that the first coordinate of $c\left( x \right)$ is 1.
	
	\begin{enumerate}
		\item For the subcase of 
		$$\left( { \star ,\; \star ,\; \star ,{0_{s_1}},\; \star ,{0_{s_2}}} \right),$$
		without loss of generality, we denote
		$$c\left( x \right) = 1 + {a_1}x + {a_2}{x^2} + {a_3}{x^t},$$ 
		where ${a_i} \in {\Bbb F}_p^{\text{*}}$ for any $0 \leqslant i \leqslant 3$ and $4 \leqslant t \leqslant 4p - 2$. 
		
		When $t$ is even, since  $c\left( 1 \right) = c\left( { - 1} \right) = 0$, it can be verified that
		\begin{equation*}  
			\left\{  
			\begin{array}{lr}  
				{1 + {a_1} + {a_2} + {a_3} = 0}, &  \\  
				{1 - {a_1} + {a_2} + {a_3} = 0}.  &    
			\end{array}  
			\right.  
		\end{equation*}
		Then one can derive that $ 2{a_1}=0 $, this contradicts ${a_1} \in {\Bbb F}_p^{\text{*}}$.
		
		When $t$ is odd, it follows from $c'\left( 1 \right) = c'\left( { - 1} \right) = 0$ that
		\begin{equation*}  
			\left\{  
			\begin{array}{lr}  
				{ {a_1} + 2{a_2} + t{a_3} = 0}, &  \\  
				{ {a_1} - 2{a_2} + t{a_3}  = 0},  &    
			\end{array}  
			\right.  
		\end{equation*}
		which impilies $ 4{a_1}=0 $. This contradicts with ${a_1} \in {\Bbb F}_p^{\text{*}}$.
		
		\item  For the subcase of 
		$$ \left( { \star ,\; \star ,{0_{s1}}, \star ,\; \star ,{0_{s2}}} \right),$$
		Without loss of generality, denote
		$$c\left( x \right) = 1 + {a_1}x + {a_2}{x^t} + {a_3}{x^{t+1}},$$ 
		where ${a_i} \in {\Bbb F}_p^{\text{*}}$ for any $0 \leqslant i \leqslant 3$ and $4 \leqslant t \leqslant 4p - 3$.
		
		When $t$ is odd, it follows from $c'\left( 1 \right) = c'\left( { - 1} \right) = 0$ that
		\begin{equation*}  
			\left\{  
			\begin{array}{lr}  
				{ {a_1} + t{a_2} + (t+1){a_3} = 0}, &  \\  
				{ {a_1} + t{a_2} - (t+1){a_3} = 0},  &    
			\end{array}  
			\right.  
		\end{equation*}
		which impilies $ t=2p-1 $. However, $c''\left( 1 \right) =0$ induces that
		$$ t(t-1){a_2} - t(t+1){a_3}  = 0.$$ Together with $ t=2p-1 $, one can immediately get $ (2p-1)(2p-2)=0 $. This is impossible, since $p$ is an odd prime.

		When $t$ is even, it follows from $c'\left( 1 \right) = c'\left( { - 1} \right) = 0$ that
		\begin{equation*}  
			\left\{  
			\begin{array}{lr}  
				{ {a_1} + t{a_2} + (t+1){a_3} = 0}, &  \\  
				{ {a_1} - t{a_2} + (t+1){a_3} = 0},  &    
			\end{array}  
			\right.  
		\end{equation*}
		which impilies $ t=2p $.
		
		Since  $c\left( 1 \right) = c\left( { - 1} \right) = 0$, it can be verified that
		\begin{equation*}  
			\left\{  
			\begin{array}{lr}  
				{1 + {a_1} + {a_2} + {a_3} = 0}, &  \\  
				{1 - {a_1} + {a_2} - {a_3} = 0}.  &    
			\end{array}  
			\right.  
		\end{equation*}
		Then one can derive that $ {a_1}+{a_3}=0 $. 
		However, $c\left( \omega \right) = c\left( { - \omega} \right) = 0$ induces that
		\begin{equation*}  
			\left\{  
			\begin{array}{lr}  
				{1 + {a_1}\omega + {a_2}\omega^{t} + {a_3}\omega^{t+1} = 0}, &  \\  
				{1 - {a_1}\omega + {a_2}\omega^{t} - {a_3}\omega^{t+1} = 0},  & 
			\end{array}  
			\right.  
		\end{equation*}
		which impilies $ {a_1}\omega  + {a_3}\omega^{t+1} = 0 $. By substituting $ t=2p $ and $ {a_1}+{a_3}=0 $, one can derive that ${a_1}={a_3}=0 $, which is contradicts ${a_i} \in {\Bbb F}_p^{\text{*}} $.

	\end{enumerate}
Through the above proof, we conclude that there is no codeword with symbol-pair weight of 6 in $\mathcal{C}_{2}$.	
Therefore, to prove that $\mathcal{C}_{2}$ is an AMDS symbol-pair code with minmum symbol-pair distance 8, it is sufficient to verify that there does not exist a codeword in $\mathcal{C}_{2}$ with symbol-pair weight 7. We have three cases to discuss.

	\item[\textbf{Case III.}] If there is a codeword in $\mathcal{C}_{2}$ with Hamming weight 6 and symbol-pair weight 7, then its certain cyclic shift must have the form
	$$ \left( { \star ,\; \star ,\; \star ,\; \star ,\; \star ,\; \star ,{0_s}} \right),\;$$
	where each $ \star $ denotes an element in  ${\Bbb F}_p^{\text{*}}$ and ${0_s}$ is all-zero vector. Without loss of generality, suppose that the first coordinate of $c\left( x \right)$ is 1. We denote that
	$$c\left( x \right) = 1 + {a_1}x + {a_2}{x^2} + {a_3}{x^3} + {a_4}{x^4} + {a_5}{x^5},$$
	this leads to $\deg\left( {c\left( x \right)} \right) = 5 < 7 = \deg\left( {g\left( x \right)} \right)$.
	
	\noindent Therefore, it shows that there does not exist a codeword in $\mathcal{ C}_{2}$ with Hamming weight 6 and symbol-pair weight 7.
	
	\item[\textbf{Case IV.}] Assume that $c\left( x \right)$ has the Hamming weight 5 and the symbol-pair weight 7. Then its certain cyclic shift must have one of the forms
	$$ \left( { \star ,\; \star ,\; \star ,{0_{s_1}}, \star ,\; \star ,{0_{s_2}}} \right)$$ or 
	$$ \left( { \star ,\; \star ,\; \star , \star ,\;{0_{s_1}}, \,\star ,{0_{s_2}}} \right),$$
	where each $ \star $ denotes an element in ${\Bbb F}_p^{\text{*}}$ and ${0_{s_1}},{0_{s_2}}$ are all-zero vectors. Without loss of generality, suppose that the first coordinate of $c\left( x \right)$ is 1.
	\begin{enumerate}

		\item  For the subcase of 
		$$ \left( { \star ,\; \star ,\; \star ,{0_{s_1}}, \star ,\; \star ,{0_{s_2}}} \right),$$
		without loss of generality, we denote
		$$c\left( x \right) = 1 + {a_1}x + {a_2}{x^2} + {a_3}{x^l} + {a_4}{x^{l + 1}},$$ 
		where ${a_i} \in {\Bbb F}_p^{\text{*}}$ for any $1 \leqslant i \leqslant 4$ and $5 \leqslant l \leqslant 4p - 3$. 
		
		\noindent When $l$  even, it can be verified that
		\begin{equation*}  
			\left\{  
			\begin{array}{lr}  
				{a_1} + 2{a_2} + l{a_3} + (l+1){a_4}  = 0, &\\
				{a_1} - 2{a_2} - l{a_3} + (l+1){a_4}  = 0, &    
			\end{array}  
			\right.  
		\end{equation*}
		since $ c'\left( 1  \right) = c'\left( { - 1 } \right) = 0$. Then one can derive that $ l=2p $.
		
		However, it follows from $ c''\left( 1  \right) =0 $ that
		$$ 2{a_2} + l(l-1){a_3} + l(l+1){a_4}  = 0. $$
		By substituting $ l=2p $, one can obtain that  $ 2{a_1}=0 $, which contradicts $ a_{1} \ne 0 $.
		
		\noindent	 When $l$  odd, it can be verified that
		\begin{equation*}  
			\left\{  
			\begin{array}{lr}  
				1+{a_1} + {a_2} + {a_3} + {a_4} = 0, &  \\  
				1-{a_1} + {a_2} - {a_3} + {a_4} = 0, &\\  
				{a_1} + 2{a_2} + l{a_3} + (l+1){a_4}  = 0, &\\
				{a_1} - 2{a_2} + l{a_3} - (l+1){a_4}  = 0, &    
			\end{array}  
			\right.  
		\end{equation*}
		since $c\left( 1 \right) = c\left( { - 1} \right) = c'\left( 1  \right) = c'\left( { - 1 } \right) = 0$. By solving the system, one can derive that $ l=2p+1 $. This implies that
		\begin{equation*}  
			\left\{  
			\begin{array}{lr}  
				{a_2} + {a_4} = -2, &  \\  
				2{a_2} +  (l+1){a_4}  = 0, &    
			\end{array}  
			\right.  
		\end{equation*}
		then one can derive that $ 2=0 $, a contradiction.
		
		\item 	For the subcase of 
		$$ \left( { \star ,\; \star ,\; \star ,\; \star ,\;{0_{s_1}}, \star ,{0_{s_2}}} \right),$$
		without loss of generality, we denote
		$$c\left( x \right) = 1 + {a_1}x + {a_2}{x^2} + {a_3}{x^3} + {a_4}{x^{l }},$$ 
		where 
		${a_i} \in {\Bbb F}_p^{\text{*}}$ for any $0 \leqslant i \leqslant 4$ and $5 \leqslant l \leqslant 4p - 2$. 
		
		When $l$ is odd, with $5 \leqslant l \leqslant 4p - 2$ and ${a_i} \in {\Bbb F}_p^{\text{*}}$ for any $1 \leqslant i \leqslant 4$. The fact $c\left( 1 \right) = c\left( { - 1} \right) = 0$ induces that
		\begin{equation*}  
			\left\{  
			\begin{array}{lr}  
				{1 + {a_1} + {a_2} + {a_3} + {a_4} = 0,} &  \\  
				{1 - {a_1} + {a_2} - {a_3} - {a_4} = 0.}  &    
			\end{array}  
			\right.  
		\end{equation*}
		This implies that ${a_2} =  - 1$.
		However, it follows from $c\left( \omega  \right) = c\left( { - \omega } \right) = 0$ that
		\begin{equation*}  
			\left\{  
			\begin{array}{lr}  
				{1 + {a_1}\omega  - {a_2} - {a_3}\omega + {a_4}{\left( \omega  \right)^l} = 0,} &  \\  
				{1 - {a_1}\omega  - {a_2} + {a_3}\omega - {a_4}{\left( \omega  \right)^l} = 0.}  &    
			\end{array}  
			\right.  
		\end{equation*}
		This leads to ${a_2} =  1$, which contradicts ${a_2} = -1$.
		
		When $l$ is even, with $6 \leqslant l \leqslant 4p - 2$ and ${a_i} \in {\Bbb F}_p^{\text{*}}$ for any $1 \leqslant i \leqslant 4$. The fact $c\left( 1 \right) = c\left( { - 1} \right) = 0$ induces that
		\begin{equation*}  
			\left\{  
			\begin{array}{lr}  
				{1 + {a_1} + {a_2} + {a_3} + {a_4} = 0,} &  \\  
				{1 - {a_1} + {a_2} - {a_3} + {a_4} = 0.}  &    
			\end{array}  
			\right.  
		\end{equation*}
		This implies that ${a_1} + {a_3} = 0$.
		However, It follows from $c\left( \omega  \right) = c\left( { - \omega } \right) = 0$ that
		\begin{equation*}  
			\left\{  
			\begin{array}{lr}  
				{1 + {a_1}\omega  - {a_2} - {a_3}\omega + {a_4}{\left( \omega  \right)^l} = 0,} &  \\  
				{1 - {a_1}\omega  - {a_2} + {a_3}\omega + {a_4}{\left( \omega  \right)^l} = 0.}  &    
			\end{array}  
			\right.  
		\end{equation*}
		This leads to ${a_1} - {a_3} = 0$. Therefore, by using ${a_1} + {a_3} = 0$, one can immediately obtain that ${a_1} = {a_3} = 0$, which contradicts ${a_i} \in {\Bbb F}_p^{\text{*}}$ for any $1 \leqslant i \leqslant 4$.
	\end{enumerate}	
	\noindent Therefore, there does not exist codeword in $\mathcal{C}_{2}$ with  Hamming weight 5 and symbol-pair weight 7.
	
	\item[\textbf{Case V.}]	Assume that $c\left( x \right)$ has Hamming weight 4 and symbol-pair weight 7. Then its certain cyclic shift must have the form
	$$c = \left( { \star ,{\text{\;}} \star ,{0_{s_1}}, \star ,{0_{s_2}}, \star ,{0_{s_3}}} \right),$$
	where each $ \star $ denotes an element in ${\Bbb F}_p^{\text{*}}$ and ${0_{s_1}},{0_{s_2}},{0_{s_3}}$ are all-zero vectors. Without loss of generality, suppose that the first coordinate of $c\left( x \right)$ is 1,
	$$c\left( x \right) = 1 + {a_1}x + {a_2}{x^u} + {a_3}{x^v},$$ 
	where $3 \leqslant u,v \leqslant 4p - 2,\left| {u - v} \right| > 1$.
	\begin{itemize}

		\item 	If both $u$ and $v$ are even, then we have
		\begin{equation*}  
			\left\{  
			\begin{array}{lr}  
				{1 + {a_1} + {a_2} + {a_3} = 0,} &  \\  
				{1 - {a_1} + {a_2} + {a_3}  = 0,}  &    
			\end{array}  
			\right.  
		\end{equation*}
		since $c\left( 1 \right) = c\left( { - 1} \right) = 0$. It follows that $2{a_1} = 0$, which contradicts $ a_{1}\ne 0 $.
		
		\item  If both $u$ and $v$ are odd, it follows from $c\left( 1 \right) = c\left( { - 1} \right) = 0$ that
		\begin{equation*}  
			\left\{  
			\begin{array}{lr}  
				{1 + {a_1} + {a_2} + {a_3} = 0,} &  \\  
				{1 - {a_1} - {a_2} - {a_3}  = 0,}  &    
			\end{array}  
			\right.  
		\end{equation*}
		which indicates $2=0$.
		
		Without loss of generality, let $u$ is odd and $v$ is even. 
		\item When  $v \equiv 0\;(\bmod~4)$, since $c'\left( 1 \right) = c'\left( { - 1} \right) = 0$ then 	
		\begin{equation*}  
			\left\{  
			\begin{array}{lr}  
				
				{a_1} + u{a_2} + v{a_3} = 0, & \\
				{a_1} + u{a_2} - v{a_3} = 0, &    
			\end{array}  
			\right.  
		\end{equation*}		
		which implies that $4p$ is a divisor of $v$. This contradicts with the code length $ n=4p $.
		
		\item 	When $v \equiv 2\;(\bmod~4)$, since $c'\left( 1 \right) = c'\left( { - 1} \right) = 0$, we have	
		\begin{equation*}  
			\left\{  
			\begin{array}{lr}  
				
				{a_1} + u{a_2} + v{a_3} = 0, & \\
				{a_1} + u{a_2} - v{a_3} = 0. &    
			\end{array}  
			\right.  
		\end{equation*}		
		this means that $2p$ is a divisor of $v$. Together with $$ u\left( {u - 1} \right){a_2} + v\left( {v - 1} \right){a_3} = 0,$$ we have $ u\left( {u - 1} \right){a_2} =0.$
		
		If $ u=kp $ for some positive integers $k$, we can obtain $ {a_1}=0 $ by $c'\left( 1 \right) = c'\left( { - 1} \right) = 0$, which is impossible. 
		
		If $ u-1=kp $ for some positive integers $k$, $ k=2 $ can be derived from $u$ odd and the code length $4p$. Hence, we have $ u=2p+1,v=2p $. However, this contradicts with $\left| {u - v} \right| > 1$.
	\end{itemize}
	
	Therefore, there does not exist codeword in $\mathcal{C}_{2}$ with Hamming weight 4 and symbol-pair weight 7.
\end{description}

As a result, $\mathcal{C}_{2}$ is an AMDS ${\left( {4p,{\text{\;}}8} \right)_p}$ symbol-pair code.	
\end{proof}

\section{Conclusion}
In this paper, employing repeated-root cyclic codes, three new classes of MDS symbol-pair codes over ${\Bbb F}_p$
with length $ lp $  and $3p$ are provided. Theorem \ref{theorem 3.1} and Corollary \ref{theorem 3.6} give some  more general generator polynomials about MDS $(lp,5)_{p}$ and $(lp,6)_{p}$ symbol-pair codes. For length $3p$, Theorem \ref{theorem of 3p MDS} provides all MDS symbol-pair codes with $d_{p}\le 12$ and also provides all AMDS symbol-pair codes with $d_{p}< 12$.
This paper is also given two new classes of AMDS symbol-pair codes over ${\Bbb F}_p$ with length $ lp $ and $ 4p $. Theorem \ref{theorem3.9} obtains a class of AMDS $(lp,7)_{p}$ symbol-pair codes and Theorem \ref{theorem 3.10} presents a class of AMDS $(4p,8)_{p}$ symbol-pair codes. 
 


\begin{thebibliography}{99}


\bibitem{Cassuto-2010-SIT}
Y. Cassuto, M. Blaum, “Codes for symbol-pair read channels,” \emph {Proc. IEEE Int. Symp. Inf.
	Theory (ISIT)}, pp. 988-992, 2010

\bibitem{Cassuto-2011-TIT}
Y. Cassuto and M. Blaum, “Codes for symbol-pair read channels,”\emph {IEEE Trans. Inf. Theory}, vol. 57, no. 12, pp. 8011–8020, Dec. 2011.

\bibitem{Cassuto-2011-SIT}
Y. Cassuto and S. Litsyn, “Symbol-pair codes: Algebraic constructions and asymptotic bounds,” \emph{ in Proc. IEEE Int. Symp. Inf. Theory},Jul./Aug. pp. 2348–2352, 2011.

\bibitem{Chee-2012-SIT}
Y. M. Chee, H. M. Kiah, C. Wang, “Maximum distance separable symbol-pair codes,” \emph {Proc. IEEE Int. Symp. Inf. Theory (ISIT)}, pp. 2886-2890, 2012

\bibitem{Chee-2013-SIT}
Y. M. Chee, L. Ji, H. M. Kiah, C. Wang, J. Yin, “Maximum distance separable codes for
symbol-pair read channels,” \emph{IEEE Trans. Inf. Theory}, vol. 59, no. 11, pp. 7259-7267, 2013

\bibitem{Kai-2015-TIT}
X. Kai, S. Zhu, P. Li, “A construction of new MDS symbol-pair codes,” \emph{IEEE Trans. Inf.
	Theory}, vol. 61, no. 11, pp. 5828-5834, 2015

\bibitem{Li-2017-DCC}
S. Li, G. Ge, “Constructions of maximum distance separable symbol-pair codes using cyclic and constacyclic codes,” \emph{Des. Codes Cryptogr.}, vol. 84, no. 3, pp. 359-372, 2017

\bibitem{Kai-2018-NCT}
X. Kai, S. Zhu, Y. Zhao, H. Luo, Z. Chen, “New MDS symbol-pair codes from repeated-root
codes,” \emph{IEEE Commun. Letters}, vol. 22, no. 3, pp. 462-465, 2018

\bibitem{Chen-2017-CIT}
B. Chen, L. Lin, H. Liu, “Constacyclic symbol-pair codes: lower bounds and optimal con-
structions,” \emph{IEEE Trans. Inf. Theory}, vol. 63, no. 12, pp. 7661-7666, 2017

\bibitem{Dinh-2018-TIT}
H. Q. Dinh, B. T. Nguyen, A. K. Singh, and S. Sriboonchitta, “On the symbol-pair distance of repeated-root constacyclic codes of prime power lengths,” IEEE Trans. Inf. Theory, vol. 64, no. 4, pp. 2417–2430, Apr. 2018.

\bibitem{Zhao-2019-Doc}
W. Zhao, “Research on algebraic theory of several types of error-correcting codes ,” 2019.

\bibitem{Dinh-2020-TIT}
H. Q. Dinh, B. T. Nguyen, S. Sriboonchitta , “MDS Symbol-Pair Cyclic Codes of Length $2 p^ s $ over ${\Bbb F}_{p^{\text{m}}}$,” IEEE Transactions on Information Theory,  66(1): 240-262, 2020.

\bibitem{Ma-2022-DCC}
J. Ma ,  J. Luo, “MDS symbol-pair codes from repeated-root cyclic codes,” \emph{Des. Codes Cryptogr.}, vol. 90, no. 1, pp. 121-137, 2021

\bibitem{Castagnoli-1991-OIT}
G. Castagnoli, J. L. Massey, P. A. Schoeller, N. von Seemann, “On repeated-root cyclic codes,” IEEE Trans. Inf. Theory, vol. 37, no. 2, pp. 337-342, 1991


\end{thebibliography}
\end{document}